\documentclass[prodmode]{acmsmall}

\usepackage[ruled]{algorithm2e}
\renewcommand{\algorithmcfname}{ALGORITHM}
\SetAlFnt{\small}
\SetAlCapFnt{\small}
\SetAlCapNameFnt{\small}
\SetAlCapHSkip{0pt}
\IncMargin{-\parindent}

\usepackage{amsfonts}
\usepackage{graphicx}

\usepackage{url}
\usepackage{epsfig}
\usepackage{graphicx}
\usepackage{xspace}
\usepackage{subfigure}
\usepackage{float}
\usepackage{amsmath}
\usepackage{amssymb}
\usepackage{latexsym}
\usepackage[noend]{algorithmic}
\usepackage{mathrsfs}
\usepackage{wasysym}
\usepackage{color}

\long\def\comment#1{}

\newcommand{\nmc}{{\sl NMC}\xspace}
\newcommand{\bss}{{\sl BSS}\xspace}
\newcommand{\bssi}{{\sl BSS-I}\xspace}
\newcommand{\bssii}{{\sl BSS-II}\xspace}
\newcommand{\rss}{{\sl RSS}\xspace}
\newcommand{\rssi}{{\sl RSS-I}\xspace}
\newcommand{\rssii}{{\sl RSS-II}\xspace}
\newcommand{\bfs}{{\sl BFS}\xspace}
\newcommand{\fexa}{{\sl EXACT}\xspace}

\def\thepage{\arabic{page}}
\begin{document}

\markboth{R-H. Li et al.}{Estimating Node Influenceability in Social
Networks}

\title{Estimating Node Influenceability in Social Networks}
\author{Rong-Hua Li
\affil{The Chinese University of Hong Kong} Jeffrey Xu Yu \affil{The
Chinese University of Hong Kong} Zechao Shang \affil{The Chinese
University of Hong Kong}}




\newcommand{\stitle}[1]{\vspace{1ex} \noindent{\bf #1}}

\newcommand{\kw}[1]{{\ensuremath {\mathsf{#1}}}\xspace}

\begin{abstract}
Influence analysis is a fundamental problem in social network
analysis and mining. The important applications of the influence
analysis in social network include influence maximization for viral
marketing, finding the most influential nodes, online advertising,
etc. For many of these applications, it is crucial to evaluate the
influenceability of a node. In this paper, we study the problem of
evaluating influenceability of nodes in social network based on the
widely used influence spread model, namely, the independent cascade
model. Since this problem is \#P-complete, most existing work is
based on Naive Monte-Carlo (\nmc) sampling. However, the \nmc
estimator typically results in a large variance, which significantly
reduces its effectiveness. To overcome this problem, we propose two
families of new estimators based on the idea of stratified sampling.
We first present two basic stratified sampling (\bss) estimators,
namely \bssi estimator and \bssii estimator, which partition the
entire population into $2^r$ and $r+1$ strata by choosing $r$ edges
respectively. Second, to further reduce the variance, we find that
both \bssi and \bssii estimators can be recursively performed on
each stratum, thus we propose two recursive stratified sampling
(\rss) estimators, namely \rssi estimator and \rssii estimator.
Theoretically, all of our estimators are shown to be unbiased and
their variances are significantly smaller than the variance of the
\nmc estimator. Finally, our extensive experimental results on both
synthetic and real datasets demonstrate the efficiency and accuracy
of our new estimators.
\end{abstract}



\keywords{Influenceability, influence networks, independent cascade
model, Stratified sampling, uncertain graph}

\acmformat{R-H. Li, Jeffrey. Yu, Z. Shang. Estimating Node
Influenceability in Social Networks.}

\begin{bottomstuff}
The work was supported by grants of the Research Grants Council of
the Hong Kong SAR, China No. 419008 and 419109.

Author's addresses: Rong-Hua Li, Jeffery Xu Yu, and Zechao Shang,
Department of System Engineering and Engineering Management, The
Chinese University of Hong Kong, Sha Tin, N.T., Hong Kong.
\end{bottomstuff}

\maketitle

\section{Introduction}
\label{sec:intro}

Large scale online social networks (OSNs) such as Facebook and
Twitter have become increasingly popular in the last years. Users in
OSNs are able to share thoughts, activities, photos, and other
information with their friends. As a result, the OSNs become an
important medium for information dissemination and influence spread.
A fundamental problem in such OSNs is to analyze and study the
social influence among users \cite{09kddinfluencetangjie}. Important
applications of influence analysis in OSNs include influence
maximization for viral marking \cite{03kddinfluence,10kddinfluence},
finding the most influential nodes
\cite{09bookinfluence,10kddinfluencenode}, online advertising, etc.
Especially, the influence maximization problem has recently
attracted tremendous attention in research community
\cite{07kddoutbreak,09kddinfluence,10kddinfluence,12vldbdatadriveninfluence}.
For many of these applications, a very important step is to
accurately evaluate the influenceability of a node in OSNs.

The influenceability evaluation problem is based on influence spread
in a network. Generally, the influence spread in a network can be
modeled as a stochastic cascade model. In the literature, a widely
used cascade mode is the independent cascade (IC) model. In the IC
model, each node $i$ has a single chance to influence his/her
neighbor $j$ with a probability $p_{ij}$, and such ``influence
event'' is independent of the other ``influence events'' over other
nodes. Due to the independent property, the IC model can be
represented by the probabilistic graph model, where each edge in the
graph is associated with a probability and the existence of an edge
is independent of any other edges \cite{10vldbuncertaingraph}. In
this paper, we focus on the IC model and assume that the influence
probabilities of all the edges in a social network are given in
advance\footnote{Learning the influence probabilities is out of
scope of this paper. In the literature, there are some studies, such
as \cite{10wsdmlearninginfluence}, on learning the influence
probabilities in social network.}. In addition, we use the
probabilistic graph model to represent the IC model.

This problem is equivalent to calculate the expected number of nodes
in $\mathcal{G}$ that are reachable from $s$, which is known to be
\#P-complete \cite{10kddinfluence}. The existing algorithms for this
problem are based on naive Monte-Carlo sampling estimator (\nmc)
\cite{03kddinfluence,05icalpinfluence,09kddinfluence}. However, \nmc
may result in a large variance, which significantly reduces its
effectiveness. We will discuss this issue in detail in
Section~\ref{sec:standmc}.

Given the IC model and a seed node $s$, the influenceability
evaluation problem is to compute the expected influence spread by
the seed node $s$. This problem is equivalent to calculate the
expected number of nodes in a probabilistic graph $\mathcal{G}$ that
are reachable from $s$, which is known to be \#P-complete
\cite{10kddinfluence}. As a result, there is no hope to exactly
evaluate the influenceability in polynomial time unless P=\#P. The
existing algorithm for this problem is based on Naive Monte-Carlo
sampling \cite{03kddinfluence,05icalpinfluence,09kddinfluence}. As
our analysis given in Section~\ref{sec:standmc}, the Naive
Monte-Carlo (\nmc) estimator leads to a large variance, and thus it
significantly reduces the effectiveness of the estimator.
Theoretically, the \nmc estimator can achieve arbitrarily close
approximation to the exact value of the influenceability. However,
this requires a large number of samples. Since performing a
Monte-Carlo estimation needs to flip $m$ coins to determine all the
$m$ edges of the network, the \nmc estimator is extremely expensive
to get a meaningful approximation of the influenceability in large
networks. Consequently, the key issue to accelerate the \nmc
estimator is to reduce the number of samples that are needed to
achieve a good accuracy.

In order to reduce the number of samples used in the \nmc estimator,
one potential solution is to reduce its variance. In this paper, we
propose two types of the Monte-Carlo estimator, namely type-I
estimator and type-II estimator, based on the idea of stratified
sampling. All of our proposed estimators are shown to be unbiased
and their variance are significantly smaller than the variance of
the \nmc estimator. To the best of our knowledge, this is the first
work that addresses and studies the variance problem in \nmc for
influenceability evaluation problem.

To develop new type-I estimators, we devise an exact
divide-and-conquer enumeration algorithm. Our exact algorithm starts
by enumerating $r$ edges, thus resulting in $2^r$ cases. Then, for
each case the algorithm recursively enumerates another $r$ edges.
The recursion will terminate after all the $m$ edges are enumerated.
This exact algorithm has exponential time complexity to evaluate
node's influenceability. Based on the exact algorithm, we propose a
basic stratified sampling (\bss) estimator, namely \bssi estimator,
to estimate a node's influenceability. In particular, we first
select $r$ edges and determine their statuses (existence or
inexistence). Obviously, this process generates $2^r$ cases. Then,
we let each case be a stratum, and draw samples separately from each
stratum. By carefully allocating the sample size for each stratum,
we prove that the variance of the \bssi estimator is smaller than
the variance of the \nmc estimator. Interestingly, we find that our
\bssi estimator can be recursively performed in each stratum, and
thereby we propose a recursive stratified sampling estimator, namely
\rssi estimator. Since the \rssi estimator recursively reduces the
variance in each stratum, its variance is significantly smaller than
the variance of the \bssi estimator. It is important to note that
both \bssi and \rssi estimators have the same time complexity as the
\nmc estimator.

In addition to the type-I estimators (\bssi and \rssi), we further
develop two type-II estimators based on a new stratification method.
The new stratification method partitions the population into $r+1$
strata by picking $r$ edges.  In the first stratum which is denoted
by stratum $0$, we set the statuses of all the $r$ edges to ``$0$'',
which denotes the edge inexistence. In the $i$-th ($1 \leq i \leq
r$) stratum, we set the statuses of all the first $i-1$ edges to
``$0$'', the $i$-th edge to ``$1$'', which signifies the edge
existence, and the rest $r-i$ edges to ``$*$'', which denotes the
status of the edge to be determined. Based on such stratification
approach, we propose a basic stratified sampling estimator, namely
\bssii estimator. Similar to the idea of the \rssi estimator, we
develop a recursive stratified sampling estimator based on \bssii
estimator, namely \rssii estimator. We conduct extensive
experimental studies on both synthetic and real datasets, and we
show that both \rssi and \rssii estimators reduce the variance of
the \nmc estimator significantly.

Note that the stratification approach in both type-I and type-II
estimators are based on the $r$ selected edges. Thus, an
edge-selection strategy may significantly affect the performance of
the estimators. In this paper, we present two edge-selection
strategies for the proposed estimators: random edge-selection and
Breadth-First-Search (\bfs) edge-selection. The random
edge-selection is to pick $r$ unsampled edges randomly for
stratification, while the \bfs edge-selection picks $r$ unsampled
edges according to their \bfs visiting order (the \bfs starts from
the seed node $s$). In our experiments, we show that an estimator
with the \bfs edge-selection strategy significantly outperforms the
same estimator with the random edge-selection strategy.

Besides the influenceability estimation problem in social networks,
our proposed estimation methods can be applied in many other
application domains. For example, consider an application in a
communication network with link failure. Given a router $s$, it
needs to count the expected number of hosts in the network that are
reachable from $s$. Such count assists network resource planing, and
is also useful for network resource estimation, for example in P2P
networks. Our proposed algorithms can provide accurate estimators
for such application domains. In addition, our influenceability
estimation methods can be directly used to the so-called influence
function evaluation problem \cite{03kddinfluence}, in which the seed
is not only one node but a set of nodes. We can solve this problem
by adding a virtual node $s$ and link it to the set of seed nodes.
Finally, our proposed stratified sampling estimators are very
general, and can be easily used to handle uncertain graph mining
problems, such as network reliability estimation
\cite{99bookchnetreliability}, shortest path
\cite{10vldbuncertaingraph}, and reachability computation problem
\cite{11vldbreacheuncertaing}.

The rest of this paper is organized as follows. We give the problem
statement in Section~\ref{sec:problem}, and introduce the Naive
Monte-Carlo estimator in Section~\ref{sec:standmc}. We propose the
type-I and type-II estimators in Section~\ref{sec:newestimatortypeI}
and Section~\ref{sec:newestimatortypeII}, respectively. Extensive
experimental studies are reported in Section \ref{sec:experiments}.
Section \ref{sec:rlwork} discusses the related work and Section
\ref{sec:concl} concludes this work.

\section{Problem Statement}
\label{sec:problem}

We consider a social network $G=(V, E)$, where $V$ denotes a set of
nodes and $E$ denotes a set of directed edges between the nodes. Let
$n=|V|$ and $m=|E|$ be the number of nodes and edges in $G$,
respectively. In a social network, users (nodes) can perform
actions, and the actions can propagate over the network.  For
example, in Twitter, an action denotes a user posts a tweet, and the
action propagation denotes the event that the same tweet is
re-posted (retweeted) by his/her followers. In this paper, we adopt
the independent cascade (IC) model
\cite{03kddinfluence,05icalpinfluence} to model such action
propagation process. In the IC model, every edge $(u, v)$ is
associated with an influence probability $p_{uv}$
(Fig.~\ref{fig:expgraph}), which represents the probability that a
node $v$ performs an action followed by the same action taken by its
adjacent node $u$. We refer to a social network $G$ with influence
probabilities as an influence network denoted by $\mathcal{G}=(V, E,
P)$, where the set $P$ represents the set of influence
probabilities. We call a node an active node if it performs an
action.

The propagation process of the IC model unfolds in discrete steps.
More precisely, we assume that a node $v$ follows a node $u$, and at
step $t$ node $u$ performs an action $\alpha$ and node $v$ does not.
Then, node $u$ is given a single chance to influence node $v$, and
it succeeds with probability $p_{uv}$. This probability is
independent of other nodes that attempt to influence node $v$. If
node $u$ succeeds, then node $v$ will perform action $\alpha$ at
step $t+1$. In other words, node $v$ is influenced by node $u$ at
step $t+1$. It is important to note that whether $u$ succeeds or
not, it cannot make any attempts to influence $v$ again. The process
terminates when there is no new node can be influenced.

The IC model can be initiated by a single node $s$ such that the
node performs an action before any other nodes in $V \backslash
\{s\}$. The seed node $s$ models the source of influence, and it can
spread across the network following the IC model. The propagation
process is a stochastic process, after the process terminates, the
number of active nodes is a random variable. Therefore, we take the
expectation of this random variable to measure the \emph{influence
spread} of $s$, and it is denoted as $F_s(\mathcal{G})$. We refer to
the expected influence spread of $s$ (i.e.\ $F_s(\mathcal{G})$) as
the influenceability of node $s$. In this paper, we aim to evaluate
the influenceability $F_s(\mathcal{G})$ given a seed node $s$. In
the following subsection, we will give a formal definition of
$F_s(\mathcal{G})$ based on the probabilistic graph model
\cite{10vldbuncertaingraph}.

\begin{figure}[t]
\begin{center}
 \subfigure[Influence network $\mathcal{G}$]{
     \label{fig:expgraph}\includegraphics[width=0.45\columnwidth]{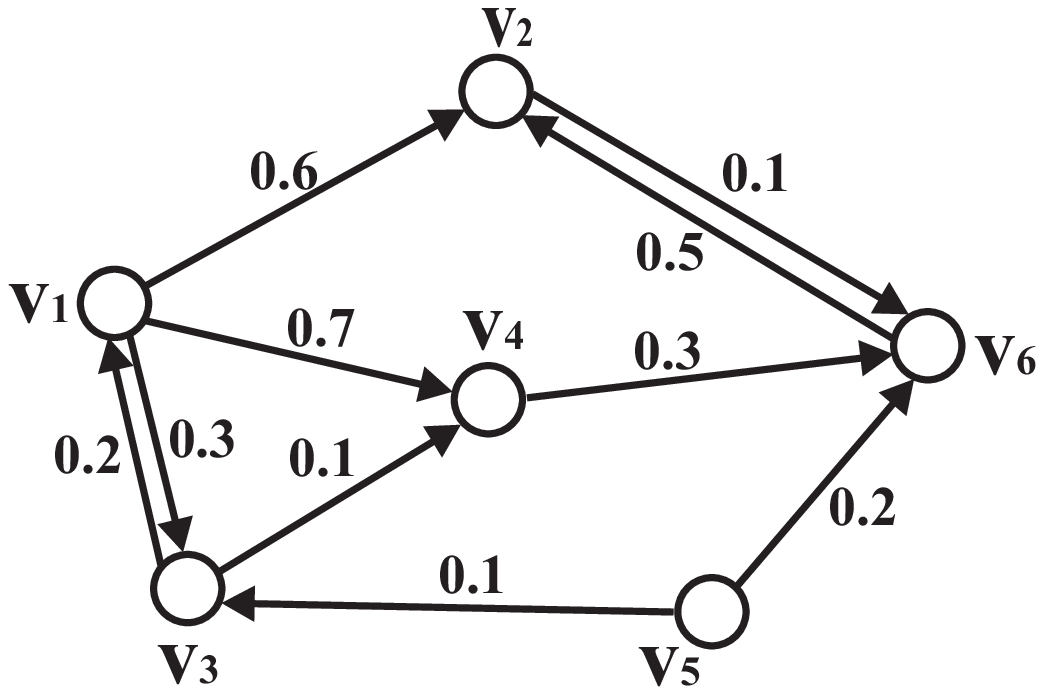}
 }
 \quad
 \subfigure[The source node $s = v_5$]{
     \label{fig:exptargetgraph}\includegraphics[width=0.45\columnwidth]{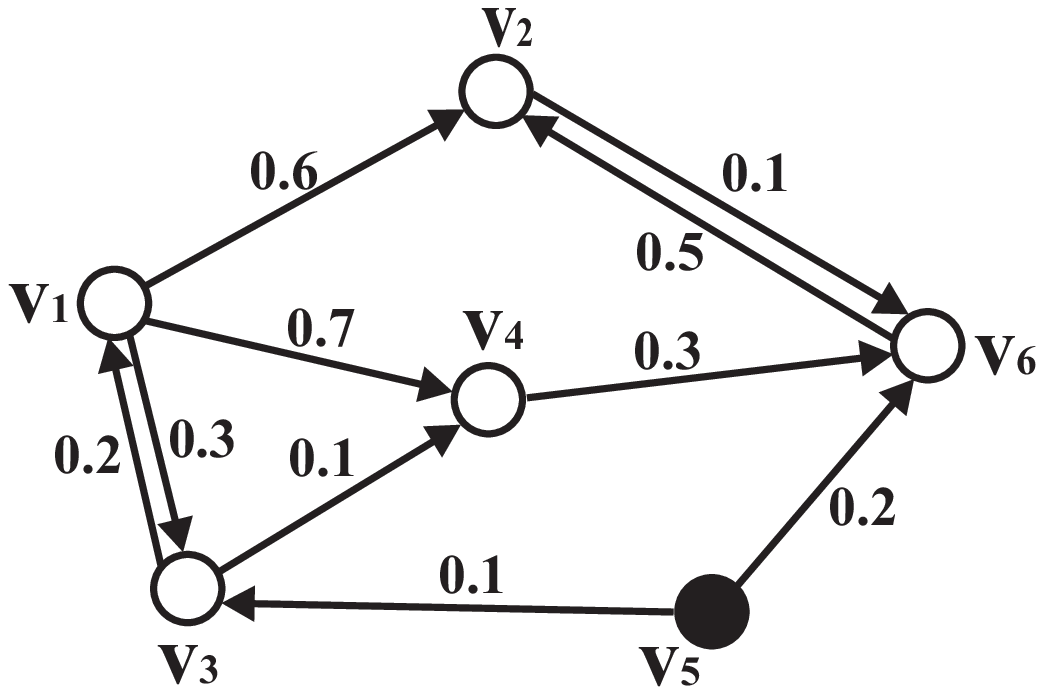}
  }
\\
\subfigure[Possible graph $G_1$ with probability $0.000007056$, and
$f_s(G_1)=3$]{
  \label{fig:exppwgraph1}\includegraphics[width=0.45\columnwidth]{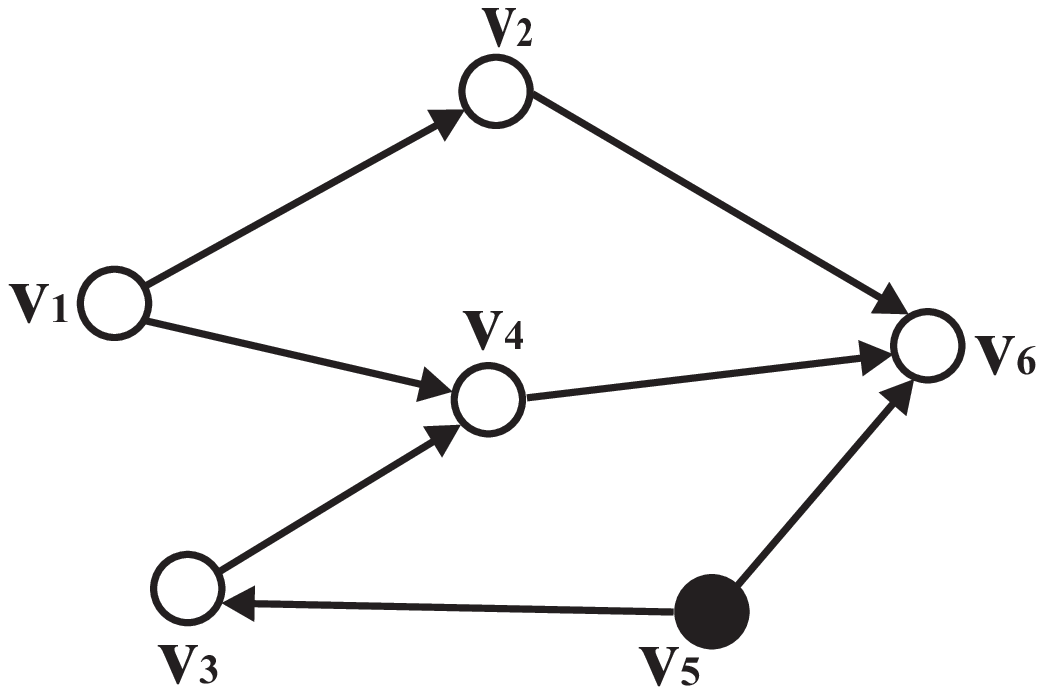}
} \quad \subfigure[Possible graph $G_2$ with probability
$0.00003704$, and
$f_s(G_2)=5$]{\label{fig:exppwgraph2}\includegraphics[width=0.45\columnwidth]{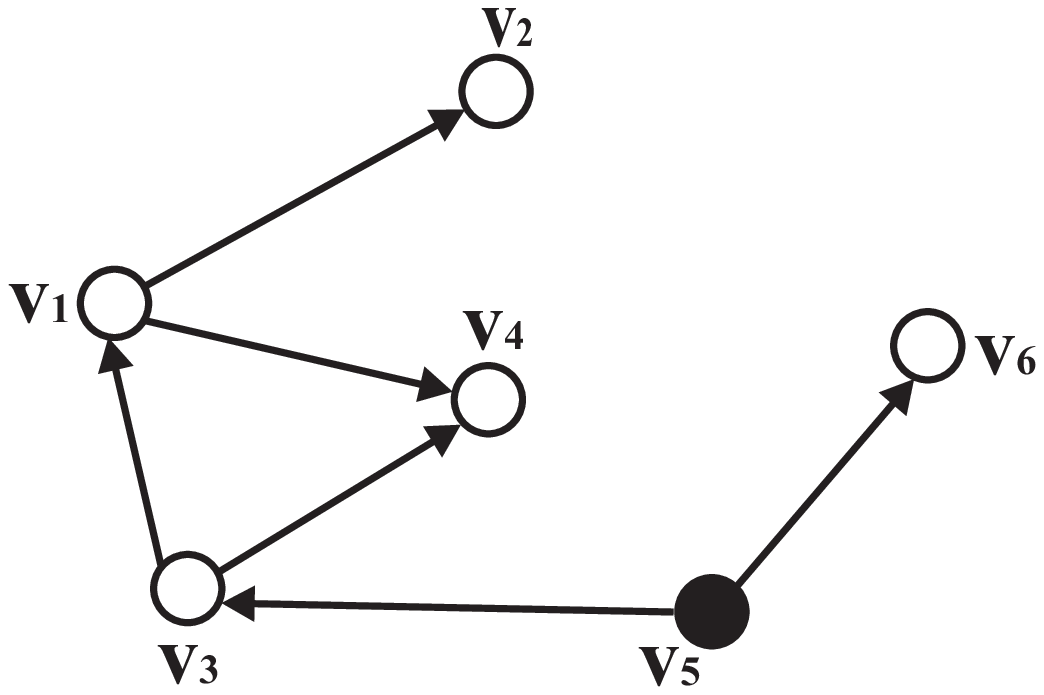}}
\end{center}
\vspace*{-0.8em} \caption[]{A Simple Influence Network}
\label{fig:examples}
\end{figure}

\subsection{Influenceability Evaluation}
\label{subsec:probdef}

Based on the IC model, an influence network $\mathcal{G}=(V, E, P)$
is represented by the probabilistic graph model
\cite{10vldbuncertaingraph}, where the existence of an edge is
independent of any other edges.

Given an influence network $\mathcal{G}=(V, E, P)$, we denote a
possible graph $G_P =(V_P, E_P)$ which can obtained by sampling each
edge $e$ in $\mathcal{G}$ according to the influence probability
$p_e$ associated with the edge $e$ ($p_e \in P$). Here, we have
$V=V_P$, $E_P \subseteq E$, and the possible graph $G_P$ has the
probability $\Pr [G_P]$, which is given by
\begin{equation}
  \label{eq:probpgdef}
\Pr [G_P] = \prod\limits_{e \in E_P } {p_e } \prod\limits_{e \in
E\backslash E_P } {(1 - p_e )}.
\end{equation}
The total number of such possible graphs is $2^m$, where $m$ is the
number of edges in ${\cal G}$. For example, in
Fig.~\ref{fig:expgraph}, the influence network $\mathcal{G}$ has
$2^{10}$ possible graphs and the possible graph $G_1$
(Fig.~\ref{fig:exppwgraph1}) and $G_2$ (Fig.~\ref{fig:exppwgraph2})
have probability $\Pr[G_1]=0.000007056$ and $\Pr[G_2]=0.00003704$,
respectively.

According to the IC model, given a seed node $s$, the
influenceability of $s$, denoted by $F_s(\mathcal{G})$, is the
expected influence spread over all the possible graphs of
$\mathcal{G}$. Therefore, based on the probabilistic graph model,
the influenceability $F_s(\mathcal{G})$ can be given by
\begin{equation}
  \label{eq:influencefundef}
  F_s(\mathcal{G}) = \sum\limits_{G_P \in \Omega } {\Pr [G_P]f_s(G_P)},
\end{equation}
where $\Omega$ denotes the set of all possible graphs of
$\mathcal{G}$, and $f_s(G_P)$ is the number of nodes that are
reachable from the seed node $s$ in the possible graph $G_P$. Note
that $f_s(G_P)$ is a random variable and its expectation is
$F_s(\mathcal{G})$, i.e. $F_s(\mathcal{G}) = \mathbb{E}[f_s(G_P)]$.

As an example, consider the source node $s = v_5$ in the influence
network ${\cal G}$ in Fig.~\ref{fig:expgraph}. $F_s(\mathcal{G})$
can be computed by enumerating all $2^{10}$ possible graphs, $G_P$,
and computing the corresponding $\Pr [G_P]$ and $f_s(G_P)$. For
instance, from Fig.~\ref{fig:exppwgraph1} and
Fig.~\ref{fig:exppwgraph2}, we have $f_s(G_1)$ $= 3$ and
$f_s(G_2)=5$. In this example, the exact $F_s(\mathcal{G})$ is
$0.46123456$.

Equipped with the definition of $F_s(\mathcal{G})$, we describe the
influenceability evaluation problem as follows.

\stitle{Problem Statement}: Given an influence network $\mathcal{G}$
and a seed node $s$, the influenceability evaluation problem is to
compute the influenceability $F_s(\mathcal{G})$
(Eq.~(\ref{eq:influencefundef})).

It is important to note that the influenceability evaluation problem
is known to be \#P-complete \cite{10kddinfluence}, even for the very
special influence network where the influence probabilities of all
edges are equivalent. There is no hope to exactly evaluate the
influenceability in polynomial time unless P = \#P. Given the
hardness of this problem, in this paper, our goal is to develop an
efficient and accurate approximate algorithm to evaluate
$F_s(\mathcal{G})$ given a seed node $s$.

An important metric for evaluating the accuracy of an approximate
algorithm is the mean squared error (MSE), which is denoted by
$\mathbb{E}[(\hat F_s(\mathcal{G}) - F_s(\mathcal{G}))^2]$, where
$\hat F_s(\mathcal{G})$ denotes an estimator of $F_s(\mathcal{G})$
by the approximate algorithm. By the so-called variance-bias
decomposition \cite{11vldbreacheuncertaing}, this metric can be
decomposed into two parts.
\begin{equation}
\label{eq:vbdecomp} \mathbb{E}[(\hat F_s(\mathcal{G}) -
F_s(\mathcal{G}))^2] = Var(\hat F_s(\mathcal{G})) + [\mathbb{E}(\hat
F_s(\mathcal{G}) - F_s(\mathcal{G}))] ^2,
\end{equation}
where $\mathbb{E}(\hat F_s(\mathcal{G}))$ and $Var(\hat
F_s(\mathcal{G}))$ denote the expectation and variance of the
estimator $\hat F_s(\mathcal{G})$, respectively. If an estimator is
unbiased, then the second term in Eq.~(\ref{eq:vbdecomp}) will be
canceled out. Therefore, the variance of the unbiased estimator
becomes the only indicator for evaluating the accuracy of the
estimator.

\section{Naive Monte-Carlo}
\label{sec:standmc}

In this section, we introduce the Naive Monte-Carlo (\nmc) sampling
for estimating the influenceability $F_s(\mathcal{G})$ given a seed
node $s$, which is the only existing algorithm used in the influence
maximization literature
\cite{03kddinfluence,07kddoutbreak,09kddinfluence,10kddinfluence}.
This method first samples $N$ possible graphs $G_1, G_2, \cdots,
G_N$ of $\mathcal{G}$ according to the influence probabilities $P$,
and then calculates the number of reachable nodes from the seed node
$s$ in each possible graph $G_i$, $i=1,2,\cdots, N$, i.e.,\
$f_s(G_i)$. Finally, the \nmc estimator $\hat F_{NMC}$ is given
below.
\begin{equation}
  \label{eq:nmcestimator}
\hat F_{NMC}  = \frac{{\sum\limits_{i = 1}^N {f_s (G_i )} }}{N}.
\end{equation}
The \nmc estimator is an unbiased estimator of $F_s(\mathcal{G})$,
such that $\mathbb{E}(\hat F_{NMC})=F_s(\mathcal{G})$. The variance
of the \nmc estimator is given as follows.
\begin{equation}
  \label{eq:nmcvariance}
  \begin{array}{l}
Var(\hat F_{NMC} ) = \frac{{\mathbb{E}[f_s (G)^2 ] - (\mathbb{E}[f_s
(G)])^2 }}{N} \\ \quad \quad \quad \quad \quad \; \;
\;=\frac{{\sum\limits_{G_P \in \Omega } {\Pr [G_P]f_s (G)^2 } - F_s
(\mathcal{G_P})^2 }}{N}.
\end{array}
\end{equation}
Notice that exactly computing the variance $Var(\hat F_{\nmc} )$ is
extremely expensive, because we have to enumerate all the possible
graphs to determine it, whose time complexity is exponential.  In
practice, we resort to an unbiased estimator of $Var(\hat F_{NMC})$
to evaluate the accuracy of the estimator $\hat F_{NMC}$
\cite{11vldbreacheuncertaing}. In this case, an unbiased estimator
of $Var(\hat F_{NMC} )$ is given by the following equation.

\begin{equation}
\label{eq:nmcvarestimator} \widetilde{Var}(\hat F_{NMC} ) =
\frac{{\sum\limits_{i = 1}^N {(f_s (G_i ) - \hat F_{NMC} )^2 } }}{{N
- 1}} = \frac{{\sum\limits_{i = 1}^N {f_s (G_i )^2 } - N\hat F_{NMC}
^2 }}{{N - 1}}.
\end{equation}

\noindent According to Eq.~(\ref{eq:nmcvarestimator}),
$\widetilde{Var}(\hat F_{NMC} )$ may be very large, because the
value of $f_s(G_i)$ falls into the interval $[0, n-1]$, which may
result in $\widetilde{Var}(\hat F_{NMC} )$ as large as $O(n^2)$.
Here, $n$ is the number of nodes in ${\cal G}$.  For example, assume
$f_s(G_i) = 0$ for $i=1, \cdots, N/2$ and $f_s(G_i) = n-1$ for
$i=N/2+1, \cdots, N$, then $\widetilde{Var}(\hat F_{NMC} )$ equals
to $N(n-1)^2/4(N-1)=O(n^2)$.  Therefore, the key issue that we
address in this paper is to design more accurate estimators than the
\nmc estimator for estimating the influenceability
$F_s(\mathcal{G})$.

The \nmc algorithm is described in Algorithm~\ref{alg:nmc}. The
algorithm works in $N$ iterations (line~2-5). In each iteration, the
\nmc algorithm needs to generate a possible graph by tossing $m$
biased coins for $m$ edges in ${\cal G}$, which takes $O(m)$ time
complexity (line~3).  Then, the algorithm invokes a \bfs algorithm
to calculate the number of reachable nodes from $s$, which again
causes $O(m)$ time complexity (line~4). As a result, the time
complexity of the \nmc algorithm is $O(Nm)$.

\begin{algorithm}[t]
\caption{\nmc($\mathcal{G}$, $N$, $s$)} \label{alg:nmc}
 {
\begin{tabbing}
    {\bf\ Input}: \hspace{0.1cm}\= Influence network $\mathcal{G}$,
               sample size $N$, and the seed node $s$. \\
{\bf\ Output}: The \nmc estimator $\hat{F}_{NMC}$.
\end{tabbing}
\begin{algorithmic}[1]
\STATE $\hat{F}_{NMC} \leftarrow 0$; \FOR {$i = 1$ to $N$} \STATE
Flip $m$ biased coins to generate a possible graph $G_i$; \STATE
Compute $f_s(G_i)$ by the \bfs algorithm; \STATE $\hat{F}_{NMC}
\leftarrow \hat{F}_{NMC} + f_s(G_i)$; \ENDFOR \STATE $\hat{F}_{NMC}
\leftarrow \hat{F}_{NMC}/N$; \STATE \textbf{return} $\hat{F}_{NMC}$;
\end{algorithmic}
}
\end{algorithm}

\comment{
\begin{table}[t]
\begin{center}
{\small
\begin{tabular}{l}
\hline
\textbf{Algorithm 1:} Naive Monte-Carlo Estimator  \\
\quad \quad \quad \quad \quad \; \; \; $\hat F$ = \nmc($\mathcal{G}$, $N$, $s$)\\
Input parameter $\mathcal{G}$: The influence network \\
Input parameter $N$: The sample size \\
Input parameter $s$: The seed node \\
Output parameter $\hat F$: The \nmc estimator \\
\hline
1: Initialize $\hat F = 0$ \\
2: \textbf{for} $i = 1$ to $N$ \textbf{do} \\
3: \quad Flip $m$ biased coins to generate a possible graph $G_i$ \\
4: \quad Compute $f_s(G_i)$ by the \bfs algorithm \\
5: \quad $\hat F$ = $\hat F$ + $f_s(G_i)$ \\
6: \textbf{end for} \\
7: $\hat F$ = $\hat F/N$ \\
8: \textbf{return} $\hat F$ \\
\hline
\end{tabular}
} \label{alg:nmc}
\end{center}
\vspace*{-0.5cm}
\end{table}
}

\section{New Type-I Estimators}
\label{sec:newestimatortypeI}

In this section, we first introduce an exact algorithm for computing
the influenceability $F_s(\mathcal{G})$, which will guide us to
design the new estimators. We will propose two new estimators based
on the idea of stratified sampling \cite{02booksampling}. Both
estimators are shown to be unbiased, and their variance is
significantly smaller than the variance of the \nmc estimator. We
refer to the two estimators as the type-I estimators.


\subsection{An exact algorithm}
\label{subsec:exactalg}

We introduce an exact divide-and-conquer enumeration algorithm to
evaluate the influenceability for a given influence network
$\mathcal{G}=(V, E, P)$ with $n$ nodes and $m$ edges. The main idea
of our exact algorithm is described as follows.
First, the algorithm divides the entire probability space $\Omega$
(all the possible graphs) into $2^r$ different subspaces by randomly
enumerating $r$ ($r < m$) edges that have not been enumerated. Note
that $r$ is a small number (eg. $r=5$). In each subspace, the exact
algorithm recursively enumerates another $r$ edges, and this process
will terminate until all the edges are enumerated. The partition
method of the exact algorithm is described in
Table~\ref{tbl:exactidea}. In Table~\ref{tbl:exactidea}, ``$0$'',
``$1$'', and ``$*$'' denote the statuses of inexistence, existence,
and not-yet-enumerated, for the edges, respectively. Each case from
$1$, $2$, $\cdots$, to $r$ corresponds to a subspace. And $\Omega
_i$, for $i=1, 2, \cdots 2^r$, denotes the probability space of the
case $i$, which represents the set of all possible graphs in the
case $i$.

\begin{table}[t]\tbl{Probability space partition in the exact algorithm
\label{tbl:exactidea}}{
\begin{tabular}{|l | l | c|}
\hline
Edges & $e_1$ $e_2$ $e_3$ $\cdots$ $e_r$ $e_{r+1}$ $\cdots$ $e_m$ & Prob. Space\\
\hline
Case 1 & \;0 \;0 \;0\; $\cdots$ \;0 \;  $ * $ \; $\cdots$ \; $ * $ & $\Omega _1$\\
Case 2 & \;1 \;0 \;0\; $\cdots$ \;0 \;  $ * $ \; $\cdots$ \; $ * $ & $\Omega _2$\\
Case 3 & \;0 \;1  \;0\; $\cdots$ \;0 \;  $ * $ \; $\cdots$ \; $ * $ & $\Omega _3$\\
$\; \; \; \cdots$ & \quad \quad  \quad \; $\cdots$  &  $\cdots$\\
Case $2^r$ & \;1 \;1  \;1\; $\cdots$ \;1 \;  $ * $ \; $\cdots$ \; $ * $ & $\Omega _{2^r}$\\
\hline
\end{tabular}}
\end{table}

To clarify our algorithm, let $T=(e_1, e_2, \cdots, e_r)$ be the set
of selected $r$ edges, and $X_i=(X_{i,1}, X_{i,2}, \cdots, X_{i,r})$
be the status vector corresponding to the selected $r$ edges under
the case $i$, where $X_{i,j}=0$ signifies that the edge $e_j$ does
no exist , and $X_{i,j}=1$ otherwise. For example, for case $1$ in
Table~\ref{tbl:exactidea}, the status vector is $X_1=(0, 0, \cdots,
0)$, which means that all the selected $r$ edges do not exist. In
other words, all the possible graphs in $\Omega _1$ do not include
the edges in $T$. The probability of a possible graph in case $i$ is
given by
\begin{equation}
  \label{eq:piexact}
\pi _i  = \Pr[G_P \in \Omega_i] = \prod\limits_{e_j  \in T \wedge
X_{i,j} = 1} {p_j } \prod\limits_{e_j  \in T \wedge X_{i,j} = 0} {(1
- p_j )}.
\end{equation}
In addition, let $A_1$ be the set of edges that have been
enumerated, and $A_2$ be the set of edges that have not been
enumerated, such that $A_1 \cup A_2 = E$, and $A_1 \cap A_2 =
\emptyset$. Then, the influenceability of the node $s$ under the
case $i$ is defined as
\begin{equation}
  \label{eq:infcasei}
  F_s (\mathcal{G}(A_1 ,A_2 ,X_i )) = \sum\limits_{G_P \in \Omega _i }
  {f_s (G_P)\frac{{\Pr [G_P]}}{{\pi  _i }}},
\end{equation}
where $\mathcal{G}(A_1 ,A_2 ,X_i )$ denotes the set of possible
graphs in the case $i$, i.e. $\Omega_i$. According to
Eq.~(\ref{eq:infcasei}), $F_s (\mathcal{G}(A_1, A_2 ,X_i))$ denotes
the expected spread over all the possible graphs in $\Omega _i$, and
${\Pr [G_P]} / {\pi _i }$ is the probability of a possible graph
$G_P$ conditioning on it exists in $\Omega_i$. It is worth of noting
that $F_s(\mathcal{G}) = F_s(\mathcal{G}(\emptyset, E, \emptyset))$.
Based on Eq.~(\ref{eq:infcasei}), we have the following theorem.

\begin{theorem}
\label{thm:factorthm} Let $F_s (\mathcal{G}(A_1, A_2, X_i))$ be the
influenceability of the node $s$ under the case $i$ as defined in
Eq.~(\ref{eq:infcasei}), and $T$ be a set of $r$ edges randomly
selected from $A_2$. For any $T$, we have $2^r$ cases, and let $Y_j$
($j=1,\cdots, 2^r$) be the corresponding status vector. Then, we
have
  \begin{equation}
    \label{eq:factorthmeq}
    F_s (\mathcal{G}(A_1, A_2, X_i)) = \sum\nolimits_{j = 1}^{2^r}
    {\pi_j F_s(\mathcal{G}(A_1 \cup T, A_2 \backslash T, [X_i,Y_j]))},
  \end{equation}
where $[X_i, Y_j]$ is a new status vector generated by appending
$Y_j$ to $X_i$.
\end{theorem}

Based on Theorem~\ref{thm:factorthm}, we develop a recursive
enumeration algorithm described in Algorithm~\ref{alg:exact}.
Algorithm~\ref{alg:exact} first partitions the entire probability
space $\Omega$ into $2^r$ subspaces, and then the same procedure
will be recursively performed on each subspace based on
Theorem~\ref{thm:factorthm} (line 9-17 in Algorithm~\ref{alg:exact}.
The algorithm terminates until all the edges are enumerated. The
influenceability $F_s(\mathcal{G})$ can be computed by invoking
\fexa($\mathcal{G}, \emptyset, E, \emptyset, s$).

\begin{algorithm}[t]
\caption{\fexa($\mathcal{G}$, $A_1$, $A_2$, $X$, $s$)}
\label{alg:exact}
 {
\begin{tabbing}
    {\bf\ Input}: \hspace{0.3cm}\= Influence network $\mathcal{G}$,
        the set of edges that have been \\
    \> enumerated $A_1$, the Set of edges that have not been \\
    \> enumerated $A_2$, sample size $N$, and the seed node $s$. \\
{\bf\ Output}: \> The exact value of $F_s(\mathcal{G})$
\end{tabbing}
\begin{algorithmic}[1]

\IF {$A_2 = \emptyset$}
  \STATE Compute $f_s(G(V, A_1, A_2, X))$ by the \bfs algorithm;
  \STATE {\bf return} $f_s(G(V, A_1, A_2, X))$;
\ELSE
  \IF {$|A_2| < r$}
     \STATE $l \leftarrow |A_2|$;
  \ELSE
     \STATE $l \leftarrow r$;
  \ENDIF
  \STATE Select $l$ edges from $A_2$ randomly;
  \STATE Let $T$ be the set of selected edges;
  \STATE $F \leftarrow 0$;
  \FOR {$i = 1$ to $2^l$}
       \STATE Let $X_i$ be the status vector of set $T$ under the case $i$;
       \STATE Compute $\pi_i $ by Eq.~(\ref{eq:piexact});
       \STATE Append $X_i$ to $X$;
       \STATE $u_i \leftarrow$ \fexa($\mathcal{G}$, $A_1 \cup T$, $A_2
              \backslash T$, $X$, $s$);
       \STATE $F \leftarrow F + \pi_i u_i$;
   \ENDFOR
   \STATE \textbf{return} $F$;
\ENDIF
\end{algorithmic}
}
\end{algorithm}

\comment{
\begin{table}[t]
\begin{center}
\small{
\begin{tabular}{l}
\hline
\textbf{Algorithm 2:} The Exact Algorithm \\
\quad \quad \quad \quad \quad \; \; \; $F$ = \fexa($\mathcal{G}$, $A_1$, $A_2$, $X$, $s$)\\
Input parameter $\mathcal{G}$: The influence network \\
Input parameter $A_1$: Set of edges that have been enumerated \\
Input parameter $A_2$: Set of edges that have not been enumerated \\
Input parameter $X$: The status vector of the enumerated edges \\
Input parameter $s$: The seed node \\
Output parameter $F$: Exact value of $F_s(\mathcal{G})$ \\
\hline
1: \textbf{if} $A_2=\emptyset$ \textbf{then} \\
2: \quad Computing $f_s(G(V, A_1, A_2, X))$ by the \bfs algorithm \\
3: \quad \textbf{return} $f_s(G(V, A_1, A_2, X))$ \\
4: \textbf{else} \\
5: \quad \textbf{if} $|A_2| < r$ \textbf{then} \\
6: \quad \quad $l = |A_2|$ \\
7: \quad \textbf{else} \\
8: \quad \quad $l=r$ \\
9: \quad \textbf{end if} \\
10: \;Selecting $l$ edges from $A_2$ randomly \\
11: \;Let $T$ be the set of selected edges \\
11: \;Initialize $F=0$ \\
12: \;\textbf{for} $i = 1$ to $2^l$ \textbf{do} \\
13: \;\quad Let $X_i$ be the status vector of set $T$ under Case $i$\\
14: \;\quad Computing $\pi  _i $ by Eq.~(\ref{eq:piexact}) \\
15: \;\quad Appending $X_i$ to $X$ \\
16: \;\quad $u_i$ = \fexa($\mathcal{G}$, $A_1 \cup T$, $A_2 \backslash T$, $X$, $s$)\\
17: \;\quad $F = F + \pi_i u_i$ \\
18: \;\textbf{end for} \\
19: \;\textbf{return} $F$\\
20: \textbf{end if} \\
\hline
\end{tabular}
} \label{alg:exactalg}
\end{center}\vspace*{-0.5cm}
\end{table}
}

The enumeration procedure given in Algorithm~\ref{alg:exact} can be
characterized by a full $2^r$-ary tree which is depicted in
Fig.~\ref{fig:enumtree}.  Note that, to simplify our analysis, here
we assume that $r$ is divisible by $m$. In the tree, each node
represents a probability space that consists of a set of possible
graphs.  For example, the root node denotes the probability space
that includes the set of all possible graphs, and each leaf node
denotes the probability space that includes only one possible graph.
Each internal node has $2^r$ children, and each child corresponds to
a case described in Table~\ref{tbl:exactidea}. To compute
$F_s(\mathcal{G})$, we need to traverse all the nodes in the tree.
Because the number of nodes in the tree is $O(2^m)$, the time
complexity of Algorithm~\ref{alg:exact} is $O(2^m)$. Therefore, the
exact algorithm only works on small networks due to the nature of
\#P-complete of the influenceability evaluation problem. In the
following, we will develop two types of efficient approximation
algorithms for evaluating the influenceability.

\begin{figure}[t]
\begin{center}
\includegraphics[width=0.6\hsize]{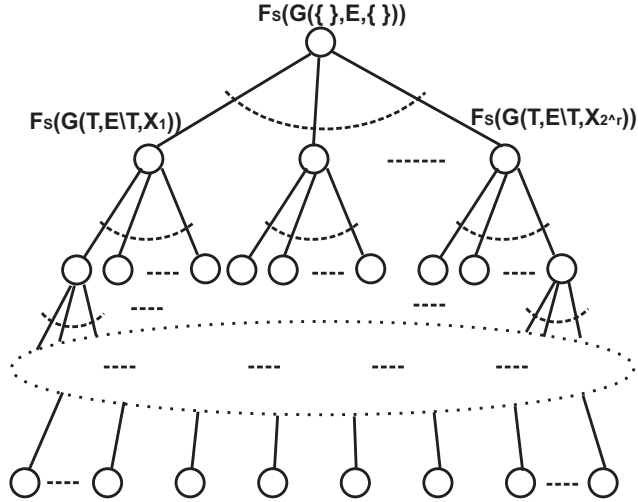}
\caption{The Enumeration Tree of the Exact
Algorithm.} \label{fig:enumtree}
\end{center}
\end{figure}

\subsection{Basic stratified sampling estimator (I)}
\label{subsec:bss}

As discussed in Section \ref{sec:standmc}, the \nmc estimator leads
to a large variance. To reduce the variance, we propose a new
stratified sampling estimator for influenceability evaluation. We
call this new estimator the basic stratified sampling (\bss)
estimator, because it servers as the basis for designing recursive
stratified sampling (\rss) estimator which will be described in
Section \ref{subsec:rss}. To distinguish the type-II estimators
which will be introduced in Section~\ref{sec:newestimatortypeII}, we
refer to the new estimators presented in this section as the type-I
estimators. Specifically, we refer to the type-I \bss and \rss
estimator as the \bssi and \rssi estimator, respectively.

Unlike the \nmc sampler which draws a sample (a possible graph) from
the entire population (all the possible graphs), the stratified
sampling \cite{02booksampling} first divides the population into $M$
disjoint groups, which are called \emph{strata}, and then
independently picks separate samples from these groups. Stratified
sampling is a commonly used technique for reducing variance
\cite{02booksampling} in sampling design. There are two key
techniques in stratified sampling: \emph{stratification}, which is a
process for partitioning the entire population into disjoint strata,
and \emph{sample allocation}, which is a procedure to determine the
sample size that needs to be drawn from each stratum. Below, we will
introduce our stratification and sample allocation method.

\stitle{Stratification}: Our idea of stratification is based on the
exact algorithm described in the previous subsection. First, we
choose $r$ edges and determine their statuses ($0/1$), where $r$ is
a small number. Recall that this process generates $2^r$ various
cases as shown in Table~\ref{tbl:exactidea}, and thereby it
partitions the set of possible graphs $\Omega$ into $2^r$ subsets
$\Omega_1, \cdots, \Omega_{2^r}$. Second, we let each subset be a
stratum. This is because $\Omega_1, \cdots, \Omega_{2^r}$ are
disjoint sets and $\Omega = \bigcup\nolimits_{i = 1}^{2^r } {\Omega
_i}$, thus each case is indeed a valid stratum. It is worth of
mentioning that our stratification process corresponds to the top
two layers in the enumeration tree (Fig.~\ref{fig:enumtree}), the
root node denotes the entire population, and each child represents a
stratum. The stratification process is depicted in
Table~\ref{tbl:statumdesign}.

\begin{table}[t]
\tbl{Stratum design of the \bssi/\rssi estimator
\label{tbl:statumdesign}}{
\begin{tabular}{|l | l | c|}
\hline Edges &  $e_1$ $e_2$ $e_3$ $\cdots$ $e_r$ $e_{r+1}$ $\cdots$
$e_m$ &
Prob. space\\
\hline Stratum 1 & \;0 \;0 \;0\; $\cdots$ \;0 \;  $ * $ \; $\cdots$
\; $ * $
& $\Omega _1$\\
Stratum 2 & \;1 \;0 \;0\; $\cdots$ \;0 \;  $ * $ \; $\cdots$ \; $ *
$
& $\Omega _2$\\
Stratum 3 & \;0 \;1  \;0\; $\cdots$ \;0 \;  $ * $ \; $\cdots$ \; $ *
$ & $\Omega _3$\\
$\; \; \; \cdots$ & \quad \quad  \quad \; $\cdots$  &  $\cdots$ \\
Stratum $2^r$ & \;1 \;1  \;1\; $\cdots$ \;1 \;  $ * $ \; $\cdots$ \;
$
* $ & $\Omega _{2^r}$\\
\hline
\end{tabular}}
\end{table}

In our stratification approach, a question that arises is how to
select the $r$ edges for stratification. As shown in our
experiments, the edge-selection strategy for choosing $r$ edges
significantly affects the performance of the estimator. One
straightforward strategy is to randomly pick $r$ edges from the edge
set $E$. We refer to this edge selection strategy as the random
edge-selection strategy.  With this strategy, the selected $r$ edges
may not have direct contributions for computing the
influenceability. For example, in Fig.~\ref{fig:exptargetgraph}, for
the source node $s = v_5$, assume $r=2$ and the selected edges are
$\{v_1 \to v_2, v_6 \to v_2\}$. The edges $\{v_1 \to v_2, v_6 \to
v_2\}$ have no direct contributions for calculating the
influenceability $F_s(\mathcal{G})$. This may reduce the performance
of the \bssi estimator.
For avoiding such a problem, we introduce another heuristic
edge-selection strategy based on the \bfs visiting order of the
edges. To estimate $F_s(\mathcal{G})$, we first perform a \bfs
algorithm starting from the node $s$ to obtain the first $r$ edges
according to the \bfs visiting order of the edges. Then, we use
these $r$ edges for stratification.  We refer to such edge-selection
strategy as the \bfs edge-selection strategy. Consider the same
example in Fig.~\ref{fig:exptargetgraph}, assume $r = 2$, the first
$r$ edges are $\{v_5 \to v_3, v_5 \to v_6\}$. Then, we partition the
population into 4 strata according to the statuses of these two
edges. Obviously, according to the \bfs edge-selection strategy, the
selected edges have direct contribution to calculate the
influenceability. In our experiments, we find that the performance
of the \bssi estimator with \bfs edge-selection strategy is
significantly better than the performance of the \bssi estimator
with random edge-selection strategy.

\stitle{The \bssi estimator}: Let $N$ be the total number of
samples, $N_i$ be the number of samples drawn from the stratum $i$
($i=1, 2, \cdots, 2^r$), and $G_{i,j}$ ($j=1,2, \cdots, N_i$) be a
possible graph sampled from the stratum $i$. Then, the \bssi
estimator is given as follows.
\begin{equation}
  \label{eq:bssestimator}
\hat F_{BSSI}  = \sum\nolimits_{i = 1}^{2^r } {\pi _i \frac{1}{{N_i
}}\sum\nolimits_{j = 1}^{N_i } {f_s(G_{i,j} )} },
\end{equation}
where $\pi_i$ is defined in Eq.~(\ref{eq:piexact}). The following
theorem shows that $\hat F_{BSSI}$ is an unbiased estimator of the
influenceability $F_s(\mathcal{G})$.

\begin{theorem}
  \label{thm:bssunbiased}
$F_s (\mathcal{G}) = \mathbb{E}(\hat F_{BSSI})$.
\end{theorem}

 \begin{proof}
We prove it by the following equalities.
\begin{equation*}
\begin{array}{l}
 \mathbb{E}(\hat F_{BSSI} ) = \mathbb{E}(\sum\nolimits_{i = 1}^{2^r } {\pi _i
 \frac{1}{{N_i }}\sum\nolimits_{j = 1}^{N_i } {f_s(G_{i,j} )} } ) \\
 \quad \;  \quad \quad \quad \; = \sum\nolimits_{i = 1}^{2^r } {\pi _i \mathbb{E}(f_s(G_{i,j} ))}  \\
  \quad \;  \quad \quad \quad \;  = \sum\nolimits_{i = 1}^{2^r } {\pi
 _i \sum\nolimits_{G_P \in \Omega _i } {f_s (G_P)\frac{{\Pr [G_P]}}{{\pi _i
 }}} }  \\
 \quad \;  \quad \quad \quad \; = \sum\nolimits_{G_P \in \Omega }
 {\Pr [G_P]f_s (G_P)}  \\
\quad \;  \quad \quad \quad \; = F_s (\mathcal{G}) \\
 \end{array}
\end{equation*}
 \end{proof}

Let $\sigma _i$ be the variance of the sample in the stratum $i$.
Since the samples are independently drawn by the basic stratified
sampling algorithm, thus the variance of the \bssi estimator is
given by
\begin{equation}\label{eq:varbss}
Var(\hat F_{BSSI} ) = \sum\nolimits_{i = 1}^{2^r } {\pi _i^2
\frac{{\sigma _i }}{{N_i }}},
\end{equation}
where $\pi_i$ is given in Eq.~(\ref{eq:piexact}).


\stitle{Sample allocation}: As discussed above, the \bssi estimator
is unbiased and the variance of the \bssi estimator depends on the
sample size of all the strata, i.e.,\ $N_i$, for $i = 1, 2, \cdots
2^r$. Thus, a question that arises is how to allocate the sample
size for each stratum $i$ ($i=1, 2, \cdots, 2^r$) to minimize the
variance of the \bssi estimator, i.e.\ $Var(\hat F_{BSSI} )$.
Formally, the sample allocation problem is formulated as follows.
\begin{equation}
\label{eq:optimalsaprob}
\begin{array}{l}
 \min \; Var(\hat F_{BSSI} ) = \sum\nolimits_{i = 1}^{2^r } {\pi _i^2 \frac{{\sigma _i }}{{N_i }}}  \\
 s.t.\quad \quad \sum\nolimits_{i = 1}^{2^r } {N_i }  = N. \\
 \end{array}
 \end{equation}
By applying the Lagrangian method, we can derive the optimal sample
allocation as given by
\begin{equation}
 \label{eq:optallocation}
 N_i  = {N\pi _i \sqrt {\sigma _i } }/{\sum\nolimits_{i = 1}^{2^r }
 {\pi _i \sqrt {\sigma _i}
 }},
 \end{equation}
for $i=1, \cdots, 2^r$. From Eq.~(\ref{eq:optallocation}), the
optimal allocation needs to know the variance of the sample in each
stratum, i.e.\ $\sigma_i$, for $i=1,\cdots, 2^r$. However, such
variances are unavailable in our problem. Interestingly, we find
that, if the sample size of the stratum $i$ is allocated to $\pi_i
N$, then the variance of the \bssi estimator will be smaller than
the variance of the \nmc estimator. We have the following theorem.

 \begin{theorem}
   \label{thm:propallocate}If $N_i = \pi_i N$, then $Var(\hat F_{BSSI}) \le Var(\hat
   F_{\nmc})$.
 \end{theorem}

\begin{proof}
If $N_i = \pi_i N$, then we have $Var(\hat F_{BSSI} ) =
\sum\nolimits_{i = 1}^{2^r } {\pi _i \frac{{\sigma _i }}{N}} $. Let
$\mu_i = \mathbb{E}(f_s (G_{i,j} ))$ be the expectation of the
sample in the stratum $i$.  By definition, we have $ \sigma _i =
\mathbb{E}(f_s (G_{i,j} )^2 ) - \mu _i ^2 = \sum\nolimits_{G_P \in
\Omega _i } {f_s (G_P)^2 \frac{{\Pr [G_P]}}{{\pi _i }}} - \mu _i ^2
$. Then, we have
 \[
\begin{array}{l}
 Var(\hat F_{BSSI} ) = \frac{1}{N}\sum\nolimits_{i = 1}^{2^r } {\pi _i
  (\sum\nolimits_{G_P \in \Omega _i } {f_s (G_P)^2 \frac{{\Pr [G_P]}}{{\pi
  _i }}}  - \mu _i ^2 )}  \\
  \quad \quad \quad \quad \quad \;\;\; = \frac{1}{N}\sum\nolimits_{i
  = 1}^{2^r } {(\sum\nolimits_{G_P \in \Omega _i } {f_s (G_P)^2 \Pr [G_P]}
  - \pi _i \mu _i ^2 )}  \\
  \quad \quad \quad \quad \quad \;\;\; = \frac{1}{N}\sum\nolimits_{G_P
  \in \Omega } {\Pr [G_P]f_s (G_P)^2 }  - \frac{1}{N}\sum\nolimits_{i =
  1}^{2^r } {\pi _i \mu _i ^2 }.  \\
 \end{array}
\]
Given this, we can derive the difference between $Var(\hat
F_{BSSI})$ and $Var(\hat F_{\nmc})$ (Eq.~(\ref{eq:nmcvariance})) as
follows:
\[
\begin{array}{l}
 \quad \quad  \; Var(\hat F_{\nmc} ) - Var(\hat F_{BSSI} ) \\
 \quad \quad \quad \; = \frac{1}{N}(\sum\nolimits_{i = 1}^{2^r } {\pi
 _i \mu _i ^2 }  - (\mathbb{E}[f_s (G_P)])^2 ) \\
  \quad \quad \quad \;  = \frac{1}{N}(\sum\nolimits_{i = 1}^{2^r } {\pi _i \mu _i ^2 }  - (\sum\nolimits_{G_P \in \Omega } {\Pr [G_P]f_s } (G_P))^2 ) \\
  \quad \quad \quad \;  = \frac{1}{N}(\sum\nolimits_{i = 1}^{2^r } {\pi _i \mu _i ^2 }  - (\sum\nolimits_{i = 1}^{2^r } {\pi _i \sum\limits_{G_P \in \Omega _i } {\frac{{\Pr [G_P]}}{{\pi _i }}f_s } (G_P)} )^2 ) \\
  \quad \quad \quad \;  = \frac{1}{N}(\sum\nolimits_{i = 1}^{2^r } {\pi _i \mu _i ^2 }  - (\sum\nolimits_{i = 1}^{2^r } {\pi _i \mu _i } )^2 ) \\
  \quad \quad \quad \; = \frac{{Var(\mu _i )}}{N} \\
  \quad \quad \quad \;  \ge 0. \\
 \end{array}
\]
Note that in the last equality $\mu_i$ can be treated as a random
variable. Then, we have $\sum\nolimits_{i = 1}^{2^r } {\pi _i \mu _i
^2 } =\mathbb{E}(\mu_i ^2)$ and $ (\mathbb{E}(\mu_i))^2 =
(\sum\nolimits_{i = 1}^{2^r } {\pi _i \mu _i } )^2 $, thus the last
equality holds. This completes the proof.
\end{proof}

\stitle{The \bssi algorithm}: Given the stratification and sample
allocation methods, we present our basic stratified sampling
algorithm in Algorithm~\ref{alg:bssi}.
First, Algorithm~\ref{alg:bssi} selects $r$ edges to partition the
population into $2^r$ strata according to an edge-selection strategy
(line 2), either random or \bfs edge-selection. For convenience, we
refer to the \bssi estimator with random edge-selection and the
\bssi estimator with \bfs edge-selection as \bssi-{\sl RM} and
\bssi-\bfs estimator, respectively.
Second, according to our sample allocation method, the algorithm
draws $\pi_iN$ samples from the stratum $i$ (line~8-13).  Finally,
the algorithm outputs the \bssi estimator $\hat F_{BSSI}$. Notice
that it takes $O(m)$ time for both generating a possible graph $G$
and performing \bfs on $G$. Besides, the algorithm needs to draw $N$
possible graphs. Hence, the time complexity of
Algorithm~\ref{alg:bssi} is $O(mN)$, which has the same complexity
as the \nmc estimator. However, our \bssi estimator significantly
reduces the variance of the \nmc estimator. The advantages of the
\bssi estimator are twofold. On one hand, given the sample size, the
\bssi estimator is more accurate than the \nmc estimator as it has a
smaller variance. On the other hand, to achieve the same variance,
the \bssi estimator needs a smaller sample size than that of the
\nmc estimator, thus it reduces the time complexity of the sampling
process.

\begin{algorithm}[t]
\caption{\bssi($\mathcal{G}$, $N$, $s$)} \label{alg:bssi}
 {
\begin{tabbing}
    {\bf\ Input}: \hspace{0.1cm}\= Influence network $\mathcal{G}$,
               sample size $N$, and the seed node $s$. \\
{\bf\ Output}: The \bssi estimator $\hat{F}$.
\end{tabbing}
\begin{algorithmic}[1]
\STATE $\hat F \leftarrow 0$; \STATE Choose $r$ edges according to
an edge-selection strategy; \FOR{$i = 1$ to $2^r$}
   \STATE Let $X_i$ be the status vector of stratum $i$;
   \STATE Compute $\pi _i $ by Eq.~(\ref{eq:piexact});
   \STATE $N_i \leftarrow [\pi_i N]$;
   \STATE $t \leftarrow 0$;
   \FOR {$j = 1$ to $N_i$}
      \STATE Flip $m-r$ coins to determine the rest $m-r$ edges;
      \STATE Let $Y_j$ be the status vector of the rest $m-r$ edges;
      \STATE Append $X_i$ to $Y_j$ to generate a possible graph $G_j$;
      \STATE Compute $f_s(G_j)$ by the \bfs algorithm;
      \STATE $t \leftarrow t + f_s(G_j)$;
   \ENDFOR
   \STATE $t \leftarrow t/N_i$;
   \STATE $\hat F \leftarrow \hat F+\pi _i t$;
\ENDFOR \STATE \textbf{return} $\hat F$;
\end{algorithmic}
}
\end{algorithm}

\comment{
\begin{table}[t]
\begin{center}
{\small
\begin{tabular}{l}
\hline
\textbf{Algorithm 3:} \bssi Estimator  \\
\quad \quad \quad \quad \quad \; \; \; $\hat F$ = \bssi($\mathcal{G}$, $N$, $s$)\\
Input parameter $\mathcal{G}$: The influence network \\
Input parameter N: The sample size \\
Input parameter $s$: The seed node \\
Output parameter $\hat F$: The \bssi estimator \\
\hline
1: Initialize $\hat F = 0$\\
2: Choosing $r$ edges according to a edge-selection strategy\\
3: \textbf{for} $i = 1$ to $2^r$ \textbf{do} \\
4: \quad Let $X_i$ be the status vector of stratum $i$ \\
5: \quad Computing $\pi _i $ by Eq.~(\ref{eq:piexact}) \\
6: \quad $N_i$ = $[\pi_i N]$ \\
7: \quad Initialize $t = 0$ \\
8: \quad \textbf{for} $j = 1$ to $N_i$ \textbf{do} \\
9: \quad \quad Flipping $m-r$ coins to determine the rest $m-r$ edges \\
10:\;\;\;\quad Let $Y_j$ be the status vector of the rest $m-r$ edges \\
11:\;\;\;\quad Appending $X_i$ to $Y_j$ to generate a possible graph $G_j$ \\
12:\;\;\;\quad Computing $f_s(G_j)$ by the \bfs algorithm \\
13:\;\;\;\quad $t$ = $t$ + $f_s(G_j)$ \\
14:\;\;\;\textbf{end for} \\
15:\;\;\;$t$ = $t/N_i$ \\
16:\;\;\;$\hat F=\hat F+\pi _i t$\\
17: \textbf{end for} \\
18: \textbf{return} $\hat F$ \\
\hline
\end{tabular}
} \label{alg:bssalg}
\end{center}\vspace*{-0.5cm}
\end{table}
}

\subsection{Recursive stratified sampling estimator (I)}
\label{subsec:rss}

Recall that the \bssi estimator splits the entire set of possible
graphs into $2^r$ subsets, which corresponds to the top two layers
in the enumeration tree (Fig.~\ref{fig:enumtree}).  Interestingly,
we observe that the basic stratified sampling (\bssi) can be applied
into any internal nodes of the enumeration tree. Based on this
observation, we develop a recursive stratified sampling estimator,
namely \rssi estimator, which is described in
Algorithm~\ref{alg:rssi}. The \rssi estimator recursively partitions
the sample size $N$ to $N_i=\pi_i N$ ($i = 1, 2, \cdots, 2^r$) for
estimating the influenceability at the stratum $i$ (line 9-19). Note
that since the \bssi estimator is unbiased, the \rssi estimator is
also unbiased. Moreover, \rssi reduces the variance at each
partition, thus the variance of \rssi is significantly smaller than
the variance of \bssi as stated by the following theorem.

\begin{theorem}
  \label{thm:rssvar}
Let $Var(\hat F_{RSSI})$ be the variance of \rssi, then $Var(\hat
F_{RSSI}) \le Var(\hat F_{BSSI})$.
\end{theorem}

\begin{proof}
We focus on the case that \rssi only partitions the population $2^r
+ 1$ times. Similar arguments can be used to prove the case of more
partitions. 
At the first partition, \rssi splits the population into $2^r$
strata, which is equivalent to \bssi. In each stratum $i$ ($i = 1,
\cdots, 2^r$), \rssi recursively partitions it into $2 ^ r$
sub-strata. Let $\Omega_i$, $\mu_i$, $\sigma_i$ and $N_i$ be the
probability space, the expectation, the variance, and the sample
size of the stratum $i$ at the first partition, respectively. Let
$\pi_i = \Pr[G_P \in \Omega_i]$ be the probability of a sample in
stratum $i$ as defined in Eq.~(\ref{eq:piexact}). Similarly, for
each stratum $i$, we denote the probability space, the expectation,
the variance, and the sample size of the sub-stratum $k$ ($k = 1,
\cdots, 2^r$), as $\Omega_{i, k}$, $\mu_{i, k}$, $\sigma_{i, k}$,
and $N_{i, k}$, respectively. Further, we denote the probability of
a sample in a sub-stratum $k$ as $\pi_{i, k}$, i.e., $\pi_{i, k} =
\Pr[G_P \in \Omega_{i, k}]$. Then, we have $\pi_{i, k} = \pi_i
\omega_k$, where $\omega_k$ denotes the probability of a sample in
sub-stratum $k$ conditioning on it is in stratum $i$, i.e.,
$\omega_k = \Pr[G_P \in \Omega_{i, k} | G_P \in \Omega_i]$. The
\rssi estimator is given by $
 \hat F _{RSSI}  = \sum\nolimits_{i = 1}^{2^r }
{\sum\nolimits_{k = 1}^{2^r } {\pi _{i,k} \frac{1}{{N_{i,k}
}}\sum\nolimits_{j = 1}^{N_{i,k} } {f _s (G_{i,k,j} )} } }, $ where
$G_{i,k,j}$ ($j = 1, \cdots, N_{i, k}$) denotes a possible graph
sampled from the sub-stratum $k$ of the stratum $i$. Then, the
variance of \rssi is $ Var(\hat F_{RSSI}) = \sum\nolimits_{i =
1}^{2^r } {\sum\nolimits_{k = 1}^{2^r } {\frac{{\pi _{i,k}^2 \sigma
_{i,k} }}{{N_{i,k} }}} }. $ By our sample allocation strategy, we
have $N_{i, k} = N \pi_{i, k}$, thereby the variance can be
simplified to $ Var(\hat F_{RSSI}) = \sum\nolimits_{i = 1}^{2^r }
{\frac{1}{N}\sum\nolimits_{k = 1}^{2^r } {\pi _{i,k} \sigma _{i,k} }
} $ Further, by $\pi_{i, k} = \pi_i \omega_k$, we have $ Var(\hat
F_{RSSI}) = \sum\nolimits_{i = 1}^{2^r } {\frac{{\pi _i
}}{N}\sum\nolimits_{k = 1}^{2^r } {\omega _k \sigma _{i,k} } }. $ By
the proportional sample allocation, we have $Var(\hat F_{BSSI}) =
\sum\nolimits_{i = 1}^{2^r } {\frac{{\pi _i }}{N}\sigma _i }$.
Therefore, the proof is completed followed by $\sum\nolimits_{k =
1}^{2^r } {\omega _k \sigma _{i,k} } \le \sigma_i$. By definition,
we have \[\begin{array}{l}
 \sum\nolimits_{k = 1}^{2^r } {\omega _k \sigma _{i,k} }  = \sum\nolimits_{k = 1}^{2^r } {\omega _k (\mathbb{E}(f _s (G_{i,k,j} )^2 ) - \mu _{i,k}^2 )}  \\
 \quad \quad \quad \quad \;\;\;\;\;\; = \sum\nolimits_{k = 1}^{2^r } {\omega _k (\sum\nolimits_{G_P  \in \Omega _{i,k} } {\frac{{\Pr [G_P ]}}{{\pi _{i,k} }}f _s (G_P )^2 }  - \mu _{i,k}^2 )}  \\
 \quad \quad \quad \quad \;\;\;\;\;\; = \sum\nolimits_{k = 1}^{2^r } {\sum\nolimits_{G_P  \in \Omega _{i,k} } {\frac{{\Pr [G_P ]}}{{\pi _i }}f _s (G_P )^2 }  - \sum\nolimits_{k = 1}^{2^r } {\omega _k \mu _{i,k}^2 } }  \\
 \quad \quad \quad \quad \;\;\;\;\;\; = \sum\nolimits_{G_P  \in \Omega _i } {\frac{{\Pr [G_P ]}}{{\pi _i }}f _s (G_P )^2 }  - \sum\nolimits_{k = 1}^{2^r } {\omega _k \mu _{i,k}^2 }.  \\
 \end{array}\]
 Then, we have \[
 \sigma _i  - \sum\nolimits_{k = 1}^{2^r } {\omega _k \sigma _{i,k} }  = \sum\nolimits_{k = 1}^{2^r } {\omega _k \mu _{i,k}^2 }  - \mu _i^2  \\
  = \sum\nolimits_{k = 1}^{2^r } {\omega _k \mu _{i,k}^2 }  - (\sum\nolimits_{k = 1}^{2^r } {\omega _k \mu _{i,k} } )^2  \ge 0. \\
\]
This completes the proof.
\end{proof}

The \rssi algorithm terminates until the sample size becomes smaller
than a given threshold ($\tau$) or the number of unsampled edges
smaller than $r$ (line 2). When the terminative conditions of the
\rssi algorithm satisfy, we perform a naive Monte-Carlo sampling for
estimating the influenceability (line 3-7).

Similar to the \bssi estimator, the partition approach in \rssi
estimator also depends on the edge-selection strategy (line~9).
Likewise, we have two edge-selection strategies for the \rssi
estimator, either random edge-selection or \bfs edge-selection.
%
%
%
For convenience, we refer to the \rssi estimator with random
edge-selection and with \bfs edge-selection as the \rssi-{\sl RM}
and \rssi-\bfs estimator, respectively.

Reconsider the example in Fig.~\ref{fig:exptargetgraph}, the \bfs
visiting order of the edges is $\{v_5 \to v_3, v_5 \to v_6, v_3 \to
v_1, v_3 \to v_4, v_6 \to v_2, v_1 \to v_2, v_1 \to v_3, v_1 \to
v_4, v_4 \to v_6, v_2 \to v_6\}$.  Assume $r=2$, according to the
\bfs visiting order, then the \rssi-\bfs first picks edge $v_5 \to
v_3$ and $v_5 \to v_6$ for stratification, and then selects the
edges $v_3 \to v_1$ and $v_3 \to v_4$, and so on. It worth of
mentioning that we can invoke the procedure \rssi($\mathcal{G},
\emptyset, E, \emptyset, N, s$), where $s$ is the seed node, to
calculate the \rssi estimator.

We analyze the time complexity of Algorithm~\ref{alg:rssi}. For
sampling a possible graph, Algorithm~\ref{alg:rssi} needs to
traverse the enumeration tree (Fig.~\ref{fig:enumtree}) from the
root node to the terminative node. Here the terminative node is a
node in the enumeration tree where the terminative conditions of the
recursion satisfy at that node, i.e.\, $N < \tau$ or $|E_2| < r$
holds in Algorithm~\ref{alg:rssi}. Let $\bar d$ be the average
length of the path from the root node to the terminative node. Then,
by analysis, the time complexity of the algorithm at each internal
node of the path is $O(r)$. Suppose that the total number of such
paths is $K$.  Then, the algorithm takes $O(K\bar d r)$ time
complexity at the internal nodes of all the paths. Note that $K$ is
bounded by the sample size $N$, and $\bar d$ is a very small number
w.r.t.\ $N$.  More specifically, we can derive that $\bar d = O(\log
_{2^r} N)$, which is a very small number. For example, assume $r=5$
and $N=100,000$, then we can get $\bar d \approx 3.3$. For all the
terminative nodes, the time complexity of the algorithm is $O(Nm)$.
This is because the algorithm needs to sample $N$ possible graphs in
total over all the terminative nodes, and for each possible graph
the algorithm performs a \bfs to compute the influenceability which
takes $O(m)$ time complexity. Since $O(K\bar d r)$ is dominated by
$O(Nm)$, the time complexity of Algorithm~\ref{alg:rssi} is
$O(Nm+K\bar d r)=O(Nm)$.

\begin{algorithm}[t]
\caption{\rssi($\mathcal{G}$, $E_1$, $E_2$, $X$, $N$, $s$)}
\label{alg:rssi}
 {
\begin{tabbing}
    {\bf\ Input}: \hspace{0.1cm}\= Influence network $\mathcal{G}$,
            the set of sampled edges $E_1$, the set of \\
         \> unsampled edges $E_2$,  sample size $N$, and the seed node $s$. \\
{\bf\ Output}: The \rssi estimator $\hat F$.
\end{tabbing}
\begin{algorithmic}[1]
\STATE $\hat F \leftarrow 0$; \IF {$N < \tau$ or $|E_2| < r$}
   \FOR {$j = 1$ to $N$}
       \STATE Flip $|E_2|$ coins to generate a possible graph $G_j$;
       \STATE Compute $f_s(G_j)$ by the \bfs algorithm;
       \STATE $\hat F \leftarrow \hat F + f_s(G_j)$;
   \ENDFOR
   \STATE \textbf{return} $\hat F/N$;
\ELSE
   \STATE Select $r$ edges from $E_2$ according to an
        edge-selection strategy \{Random or \bfs visiting order\};
   \STATE Let $T$ be the set of selected edges;
   \FOR {$i = 1$ to $2^r$}
      \STATE $Y \leftarrow X$ \{Recording the current status vector $X$\};
      \STATE Let $X_i$ be the status vector of set $T$ in stratum $i$;
      \STATE Append $X_i$ to $Y$;
      \STATE Compute $\pi _i $ by Eq.~(\ref{eq:piexact});
      \STATE $N_i  \leftarrow \left[ {\pi _i N} \right]$;
      \STATE $\mu_i \leftarrow$ \rssi($\mathcal{G}$, $E_1 \cup T$,
      $E_2 \backslash T$, $Y$, $N_i$, $s$);
      \STATE $\hat F \leftarrow \hat F + \pi_i \mu_i$;
   \ENDFOR
   \STATE  \textbf{return} $\hat F$;
\ENDIF
\end{algorithmic}
}
\end{algorithm}

\section{New Type-II Estimators}
\label{sec:newestimatortypeII}

In this section, we propose two new stratified sampling estimators,
namely type-II basic stratified sampling (\bssii) estimator and
type-II recursive stratified sampling (\rssii) estimator. The \bssii
and \rssii are shown to be unbiased and their variance are
significantly smaller than the variance of the \nmc estimator. In
the following, we first introduce the \bssii estimator, and then
present the \rssii estimator.

\subsection{Basic stratified sampling estimator (II)}
\label{subsec:bssII} 

\stitle{Stratification}: We propose a new stratification method for
the \bssii estimator. This new stratification method splits the
entire probability space $\Omega$ into $r+1$ various subspaces
($\Omega_0, \cdots, \Omega_r$) by choosing $r$ edges. Specifically,
for stratum $0$, we set the statuses of all the $r$ selected edges
to ``0'', and for the stratum $i$ ($i \ne 0$), we set the status of
edge $i$ to ``1'' and the statuses of all the previous $i-1$ edges
(i.e.\ $e_1, \cdots, e_{i-1}$) to ``0''. Unlike the stratification
method of the \bssi estimator, this new stratification approach
allows us to set $r$ to be a big number, such as $r=50$. The stratum
design method is depicted in Table~\ref{tbl:newstatumdesign}.

\begin{table}[t] \tbl{Stratum design of the \bssii/\rssii estimator
\label{tbl:newstatumdesign}}{
\begin{tabular}{|l|l|c|}
\hline Edges &  $e_1$ $e_2$ $e_3$ $\cdots$ $e_r$ $e_{r+1}$ $\cdots$
$e_m$ & Prob. space\\ \hline
Stratum 0 & \;0 \;0  \;0 \; $\cdots$\;0 \;  $ * $ \; $\cdots$ \; $ * $ & $\Omega_0$\\
Stratum 1 & \;1 \;$ * $ \;$ * $\; $\cdots$ \;$ * $ \;  $ * $ \; $\cdots$ \; $ * $ & $\Omega_1$\\
Stratum 2 & \;0 \;1 \;$ * $\; $\cdots$ \;$ * $ \;  $ * $ \; $\cdots$ \; $ * $ & $\Omega_2$\\
Stratum 3 & \;0 \;0  \;1\; $\cdots$ \;$ * $ \;  $ * $ \; $\cdots$ \; $ * $& $\Omega_3$\\
$\; \; \; \;\cdots$ & \; \; \;\;\;\; \; \; $\cdots$  &  $\cdots$\\
Stratum $r$ & \;0 \;0  \;0\; $\cdots$ \;1 \;  $ * $ \; $\cdots$ \; $ * $ & $\Omega_{r}$\\
\hline
\end{tabular}
}
\end{table}

In Table~\ref{tbl:newstatumdesign}, each stratum (Stratum 0, Stratum
1, $\cdots$, Stratum $r$) corresponds to a subspace ($\Omega_0,
\Omega_1, \cdots, \Omega_r$). For any $i \ne j$, we have $\Omega _i
\cap \Omega _j = \phi$. Below, we show that $\bigcup\nolimits_{i =
0}^r {\Omega _i } = \Omega$. Let $T=(e_1, e_2, \cdots, e_r)$ be the
set of $r$ selected edges and $p_i$ ($i = 1$) be the corresponding
influence probability, then the probability of a possible graph in
stratum $i$ is given by
\begin{equation}
  \label{eq:nrsspi}
  \pi ^\prime _i  = \Pr[G_P \in \Omega _i ] = \left\{ \begin{array}{l}
 \prod\limits_{j = 1}^r {(1 - p_j )} ,\quad \;\; if \; i = 0 \\
 p_i\prod\limits_{j = 1}^{i - 1} {(1 - p_j ) ,\quad otherwise}  \\
 \end{array} \right.
\end{equation}
The following theorem implies $\bigcup\nolimits_{i = 0}^r {\Omega _i
} = \Omega $.

\begin{theorem}
\label{thm:stadesign} $\Pr [G_P \in \Omega ] = \sum\nolimits_{i =
0}^r {\Pr [G_P \in \Omega _i ]}  = 1$.
\end{theorem}

\begin{proof}
We prove it by the following equalities.

{\small
\[\begin{array}{l}
 \sum\nolimits_{i = 0}^r {\Pr [G_P \in \Omega _i ]}  \\
 = \prod\nolimits_{j = 1}^r {(1 - p_j )}  + p_1  + (1 - p_1 )p_2  +  \cdots  + p_r \prod\nolimits_{j = 1}^{r - 1} {(1 - p_j )}  \\
  = \prod\nolimits_{j = 1}^{r - 1} {(1 - p_j )}  + p_1  + (1 - p_1 )p_2  +  \cdots  + p_{r - 1} \prod\nolimits_{j = 1}^{r - 2} {(1 - p_j )}  \\
  \cdots  \\
  = 1 - p_1 + p_1 \\
  = 1 \\
 \end{array}\]
}
\end{proof}

Armed with Theorem \ref{thm:stadesign}, we conclude that the stratum
design approach described in Table~\ref{tbl:newstatumdesign} is a
valid stratification method.

\stitle{The \bssii estimator}: Similar to the \bssi estimator, we
let $N$ be the total sample size, and $N_i$ be the sample size of
the stratum $i$, and $G_{i,j}$ ($j=1, 2, \cdots, N_i$) be a possible
graph sampled from the stratum $i$. Then the \bssii estimator $\hat
F_{BSSII}$ is given by
\begin{equation}
\hat F_{BSSII}  = \sum\nolimits_{i = 0}^r {\pi^\prime _i
\frac{1}{{N_i }}} \sum\nolimits_{j = 1}^{N_i } {f_s (G_{i,j} )},
\end{equation}
where $\pi ^\prime _i$ is given in Eq.~(\ref{eq:nrsspi}). Similar to
Theorem \ref{thm:bssunbiased}, the following theorem shows that the
\bssii estimator is unbiased. The proof is similar to the proof of
Theorem~\ref{thm:bssunbiased}, thus we omit for brevity.

\begin{theorem}
  \label{thm:bssiiunbiased}
$F_s (\mathcal{G}) = \mathbb{E}(\hat F_{BSSII} )$.
\end{theorem}

The variance of the \bssii estimator is given by
\begin{equation}
Var(\hat F_{BSSII}) = \sum\nolimits_{i = 0}^r {{\pi ^\prime _i}^2
\frac{{\sigma _i }}{{N_i }}},
\end{equation}
where $\sigma _i$ denotes the variance of the sample in the stratum
$i$.

\stitle{Sample allocation}: Analogous to the \bssi estimator, for
the \bssii estimator, we can derive that the optimal sample
allocation is given by $N_i = N\pi ^\prime _i \sqrt {\sigma _i }
/\sum\nolimits_{i = 0}^r {\pi ^\prime _i \sqrt {\sigma _i}}$.  This
optimal allocation strategy needs to know the variance of the sample
in each stratum, which is impossible in our problem.  Therefore,
similar to the sample allocation approach used in the \bssi
estimator, for the \bssii estimator, we set the sample size of the
stratum $i$ equals to $\pi^\prime _i N$, i.e.\ $N_i = \pi^\prime _i
N$. On the basis of this sample allocation method, we show that the
variance of the \bssii estimator is smaller than the variance of the
\nmc estimator as stated by the following theorem. The proof of the
theorem is similar to theorem~\ref{thm:propallocate}, thus we
omitted for brevity.

\begin{theorem}
\label{thm:bssIIsamplealloc}If $N_i = \pi ^\prime _i N$, $Var(\hat
F_{BSSII}) \le Var(\hat F_{\nmc})$.
\end{theorem}

However, it is very hard to compare the variance of the \bssii
estimator with the variance of the \bssi estimator. In our
experiments, we find that these two estimators achieve comparable
variance.

\stitle{The \bssii algorithm}: With the stratification and sample
allocation method, we describe the \bssii algorithm in
Algorithm~\ref{alg:bssii}.  Algorithm~\ref{alg:bssii} picks $r$
edges to split the entire population into $r+1$ strata in terms of
an edge-selection strategy (line~2).  Any of the two edge-selection
strategies (random edge-selection and \bfs edge-selection) used in
the \bssi algorithm can also be used in the \bssii algorithm. We
refer to the \bssii estimator with the random edge-selection and the
\bssii estimator with \bfs edge-selection as \bssii-{\sl RM} and
\bssii-\bfs estimator, respectively.
In terms of the sample allocation method of the \bssii estimator,
Algorithm~\ref{alg:bssii} picks $N_i = \pi ^\prime _i N$ samples
from the stratum $i$, for $i = 0, 1, \cdots, r$, and outputs the
\bssii estimator $\hat F_{BSSII}$. Like the \bssi estimator, the
time complexity of \bssii estimator is $O(Nm)$. This is because the
\bssii needs to draw $N$ possible graphs, and both sampling each
possible graph $G$ and computing $F_s(G)$ take $O(m)$ time.

\begin{algorithm}[t]
\caption{ \bssii($\mathcal{G}$, $N$, $s$)} \label{alg:bssii}
 {
\begin{tabbing}
    {\bf\ Input}: \hspace{0.1cm}\= Influence network $\mathcal{G}$,
               sample size $N$, and the seed node $s$. \\
{\bf\ Output}: The \bssii estimator $\hat F$.
\end{tabbing}
\begin{algorithmic}[1]
\STATE $\hat F \leftarrow 0$; \STATE Select $r$ edges according to
an edge-selection strategy; \FOR {$i = 0$ to $r$}
   \STATE Compute $\pi ^\prime _i $ by Eq.~(\ref{eq:nrsspi});
   \STATE $N_i \leftarrow [\pi ^\prime _i N]$;
   \STATE $t \leftarrow 0$;
   \IF {i = 0}
      \STATE $k \leftarrow r$;
   \ELSE
      \STATE $k \leftarrow i$;
   \ENDIF
   \STATE Let $E_i$ be the set of edges to be determined under stratum $i$;
   \FOR {$j = 1$ to $N_i$}
      \STATE Flip $m-k$ coins to determine $E_i$, and thus generate a
             possible graph $G_j$;
      \STATE Compute $f_s(G_j)$ by the \bfs algorithm;
      \STATE $t \leftarrow t + f_s(G_j)$;
   \ENDFOR
   \STATE $t \leftarrow t/N_i$;
   \STATE $\hat F \leftarrow \hat F+\pi ^\prime _i t$;
\ENDFOR \STATE \textbf{return} $\hat F$;
\end{algorithmic}
}
\end{algorithm}

\comment{
\begin{table}[t]\vspace*{-0.5cm}
\begin{center}
{\small
\begin{tabular}{l}
\hline
\textbf{Algorithm 5:} \bssii Estimator  \\
\quad \quad \quad \quad \quad \; \; \; $\hat F$ = \bssii($\mathcal{G}$, $N$, $s$)\\
Input parameter $\mathcal{G}$: The influence network \\
Input parameter N: The sample size \\
Input parameter s: The seed node \\
Output parameter $\hat F$: The \bssii estimator \\
\hline
1: Initialize $\hat F = 0$ \\
2: Selecting $r$ edges according to a edge-selection strategy \\
3: \textbf{for} $i = 0$ to $r$ \textbf{do} \\
4: \quad Computing $\pi ^\prime _i $ by Eq.~(\ref{eq:nrsspi}) \\
5: \quad $N_i$ = $[\pi ^\prime _i N]$ \\
6: \quad Initialize $t = 0$\\
7: \quad \textbf{if} i = 0 \textbf{then} \\
8: \quad \quad $k=r$ \\
9: \quad \textbf{else} \\
10:\;\;\; \quad $k=i$ \\
11:\;\;\;\textbf{end if} \\
12:\;\;\;Let $E_i$ be the set of edges to be determined under stratum $i$\\
13:\;\;\;\textbf{for} $j = 1$ to $N_i$ \textbf{do} \\
14:\;\;\;\quad Flipping $m-k$ coins to determine $E_i$, \\
\quad \quad \quad and thus generate a possible graph $G_j$ \\
15:\;\;\;\quad Computing $f_s(G_j)$ by the \bfs algorithm \\
16:\;\;\;\quad $t$ = $t$ + $f_s(G_j)$ \\
17:\;\;\;\textbf{end for} \\
18:\;\;\;$t$ = $t/N_i$ \\
19:\;\;\;$\hat F=\hat F+\pi ^\prime _i t$\\
20: \textbf{end for} \\
21: \textbf{return} $\hat F$ \\
\hline
\end{tabular}
} \label{alg:bssalg}
\end{center}\vspace*{-0.5cm}
\end{table}
}

\subsection{Recursive stratified sampling estimator (II)}
\label{subsec:rssII}

Based on the \bssii estimator, in this subsection, we develop
another new recursive stratified sampling estimator, namely \rssii
estimator. Similar to the idea of the \rssi estimator, the \rssii
estimator makes use of the \bssii estimator as the basic component
and recursively applies the \bssii estimator at each stratum. More
specifically, the \rssii estimator first partitions the entire
probability space $\Omega$ into $r+1$ subspace $\Omega_i$ ($i=0,1,
\cdots, r$) according to the stratification method of the \bssii
estimator. The same partition procedure is recursively performed in
each subspace $\Omega_i$. At each partition, the \rssii estimator
utilizes the same sample allocation method as the \bssii estimator
to allocate the sample size. The recursion process of the \rssii
estimator will terminate until the sample size is smaller than a
given threshold ($\tau$) or the number of unsampled edges is smaller
than $r$. Since the \bssii estimator is unbiased, the \rssii
estimator is also unbiased. The variance of the \rssii estimator is
smaller than the variance of the \bssii estimator, because the
\rssii estimator recursively reduces variance at each partition
while the \bssii estimator only reduces variance at one partition.
Similar to Theorem~\ref{thm:rssvar}, we have the following theorem.
\begin{theorem}
  \label{thm:rssiivar}
Let $Var(\hat F_{RSSII})$ be the variance of \rssi, then $Var(\hat
F_{RSSII}) \le Var(\hat F_{BSSII})$.
\end{theorem}

The detail algorithm of the \rssii estimator is described in
Algorithm~\ref{alg:rssii}. Firs, according to an edge-selection
strategy, Algorithm~\ref{alg:rssii} selects $r$ edges from the
unsampled edge-set, which is denoted by $E_2$, to partition the
population into $r+1$ strata (line~9). Note that the random
edge-selection and \bfs edge-selection strategy used in the \rssi
estimator can also be applied in the \rssii estimator. We refer to
the \rssii estimator with random edge-selection and \bfs
edge-selection as the \rssii-{\sl RM} and \rssii-\bfs estimator,
respectively.
Second, according to the sample allocation method, the algorithm
recursively invokes the \rssii algorithm with sample size $N_i$ in
stratum $i$, for $i=1, \cdots, r$ (line 11-23). In line~15 and
line~20, we let $X_i$ be the status vector of the selected edges
under the stratum $i$. Unlike the \rssi estimator, the status vector
of the \rssii estimator is determined by the stratification method
of the \bssii estimator (Table~\ref{tbl:newstatumdesign}). For
example, at the first partition of the \rssii estimator, assume
$T=(e_1, e_2, \cdots, e_r)$ is the set of $r$ edges selected, the
status vector of these selected edges under the stratum $0$ is $X_0
= (0,0,\cdots,0)$. The status vector under the stratum $i$ is $X_i =
(0, \cdots, 0, 1, * \cdots, *)$, where the statuses of the first
$i-1$ edges are ``0'', the status of the $i$-th edge is ``1'', and
the rest $r-i$ edges are ``$*$''. Finally, the algorithm outputs the
\rssii estimator (line~24).

\begin{algorithm}[t]
\caption{\rssii($\mathcal{G}$, $E_1$, $E_2$, $X$, $N$, $s$)}
\label{alg:rssii}
 {
\begin{tabbing}
    {\bf\ Input}: \hspace{0.3cm}\= Influence network $\mathcal{G}$,
            the set of sampled edges $E_1$, \\
         \> the set of unsampled edges $E_2$,  sample size $N$, \\
         \> and the seed node $s$. \\
    {\bf\ Output}: \> The \rssii estimator $\hat F$.
\end{tabbing}
\begin{algorithmic}[1]
\STATE $\hat F \leftarrow 0$; \IF {$N < \tau$ or $|E_2| < r$}
   \FOR {$j = 1$ to $N$}
      \STATE Flip $|E_2|$ coins to generate a possible graph $G_j$;
      \STATE Compute $f_s(G_j)$ by the \bfs algorithm;
      \STATE $\hat F \leftarrow \hat F + f_s(G_j)$;
   \ENDFOR
   \STATE \textbf{return} $\hat F/N$;
\ELSE
   \STATE Select $r$ edges from $E_2$ according to an edge-selection
   strategy (random or \bfs visiting order);
   \STATE  Let $T=(e_1, e_2, \cdots, e_r)$ be the set of selected edges;
   \FOR {$i = 0$ to $r$}
      \STATE Compute $\pi ^\prime _i $ by Eq.~(\ref{eq:nrsspi});
      \STATE $N_i  \leftarrow \left[ {\pi ^\prime _i N} \right]$;
      \IF {$i=0$}
         \STATE Let $X_0$ be the status vector of set $T$ under stratum $0$;
         \STATE Append $X_0$ to $X$;
         \STATE $\mu_i \leftarrow$ \rssii($\mathcal{G}$, $E_1 \cup T$, $E_2
     \backslash T$, $X$, $N_i$, $s$);
      \ELSE
         \STATE Let $T_i \leftarrow \{e_1,\cdots, e_i\}$;
         \STATE Let $X_i$ be the status vector of set $T_i$ under
     stratum $i$;
         \STATE Append $X_i$ to $X$;
         \STATE $\mu_i \leftarrow$ \rssii($\mathcal{G}$, $E_1 \cup T_i$, $E_2
     \backslash T_i$, $X$, $N_i$, $s$);
      \ENDIF
      \STATE $\hat F \leftarrow \hat F + \pi ^\prime _i \mu_i$;
   \ENDFOR
   \STATE \textbf{return} $\hat F$;
\ENDIF
\end{algorithmic}
}
\end{algorithm}

Like the \rssi estimator, to sample a possible graph, the \rssii
algorithm needs to traverse the recursive tree from the root node to
the terminative node. At all the terminative nodes, the algorithm
needs to sample $N$ possible graphs in total, and for each possible
graph it needs to perform a \bfs to compute the influenceability,
thus the time complexity is $O(Nm)$. At each internal node in a path
from the root node to the terminative node, the time complexity is
$O(r)$. This is because at each internal node the algorithm only
needs to select $r$ edges and determine their statuses which consume
$O(r)$ time complexity. Let $\bar d$ be the average length of such
path and $K$ be the total number of paths. Then, for all the
internal nodes, the algorithm takes $O(K \bar d r)$ time complexity.
According to the terminative condition given in
Algorithm~\ref{alg:rssii}, we can derive that $\bar d = \min \{ \log
_r N,\log _r m)$. Since $r$ can be a big number (eg.\ $r=50$), $\bar
d$ is very small. Thus, the time complexity at the internal nodes
$O(K \bar d r)$ can be dominated by $O(Nm)$. We conclude that the
average time complexity of Algorithm~\ref{alg:rssii} is $O(Nm)$.

\section{Experiments}
\label{sec:experiments}

We conduct experimental studies for different estimators over four
datasets. We confirm the efficiency and accuracy of the proposed
estimators. In the following, we first describe the experimental
setup, and then report our results.

\subsection{Experimental setup}
\label{subsec:expsetup}

\stitle{Datasets}: We use one synthetic dataset and three real
datasets in our experiments. We apply the same parameters used in
\cite{11vldbreacheuncertaing} to generate the synthetic dataset. For
the graph topology, we generate an Erdos-Renyi (ER) random graph
with 5,000 vertices and edge density 10. For the influence
probabilities, we generate a probability for each edge according to
a [0,1] uniform distribution.

The three real datasets are given as follows. (1) FacebookLike
dataset: this dataset originates from a Facebook social network for
students at University of California, Irvine. It contains the users
who sent or received at least one message. We collect this dataset
from (\url{toreopsahl.com/datasets}). The dataset is a weighted
graph, and the weight of each edge denotes the number of messages
passing over the edge. (2) Condmat dataset: this dataset is a
weighted collaboration network, where the weight of an edge
represents the number of co-authored papers between two
collaborators. We download this dataset from
(\url{www-personal.umich.edu/~mejn/netdata}). (3) DBLP dataset: this
dataset is also a weighted collaboration network, where the weight
of the edge signifies the number of co-authored papers. This dataset
is provided by the authors in \cite{10icdmdataset}. Table
\ref{tbl:data} summarizes the information for the four real
datasets. To obtain the influence networks, for each real dataset,
we generate the influence probabilities according to the same method
used in \cite{10vldbuncertaingraph,11vldbreacheuncertaing}.
Specifically, to generate the probability of an edge, we apply an
exponential cumulative distribution function (CDF) with mean 2 to
the weight of the edge.

\begin{table}[t]\tbl{Summary of the datasets
\label{tbl:data}}{
\begin{tabular}{|l|r|r|l|} \hline
{\bf Name} & {\bf Nodes} & {\bf Edges} & {\bf Ref.} \\
\hline
Random graph & 5,000 & 50,616 & \cite{11vldbreacheuncertaing}\\
FacebookLike & 1,899 & 20,296 & \cite{09weightednet} \\
Condmat & 16,264 & 95,188 & \cite{01newmandataset} \\
DBLP & 78,648 & 376,515 & \cite{10icdmdataset}\\
\hline
\end{tabular}}
\end{table}

\stitle{Different estimators}: In our experiments, we compare 10
estimators.
(1) The \nmc estimator, which is the Naive Monte-Carlo estimator.
(2) \rssi-{\sl RM} ($r=1$), which is a special \rssi-{\sl RM}
estimator where the parameter $r=1$, based on work presented in
\cite{11vldbreacheuncertaing} for computing distance-constraint
reachability on uncertain graph. We also generalize their estimator
to arbitrary parameter $r$, and apply the generalized estimator for
influenceability evaluation.
Recall that beyond the random edge-selection strategy, we propose a
more accurate \rssi estimator with \bfs edge-selection strategy.
(3) \bssi-{\sl RM}, which is the \bssi estimator with the random
edge-selection. (4) \bssi-\bfs, which is the \bssi estimator with
the \bfs edge-selection. (5) \rssi-{\sl RM}, which is the \rssi
estimator with the random edge-selection. (6) \rssi-\bfs, which is
the \rssi estimator with the \bfs edge-selection. (7) \bssii-{\sl
RM}, which is the \bssii estimator with the random edge-selection.
(8) \bssii-\bfs, which is the \bssii estimator with the \bfs
edge-selection. (9) \rssii-{\sl RM}, which is the \rssii estimator
with the random edge-selection. (10) \rssii-\bfs, which is the
\rssii estimator with the \bfs edge-selection.

\stitle{Evaluation metric}: Two metrics are used to evaluate the
performance of the estimators: running time and relative variance.
The running time evaluates the efficiency of the estimators. The
relative variance is leveraged to evaluate the accuracy of the
estimators. Let $\sigma_{NMC}$ be the variance of the \nmc
estimator. We calculate the relative variance of an estimator $\hat
F$ by $\sigma_{\hat F}/\sigma_{NMC}$. Since computing the exact
variance of the estimators is intractable, we resort to an unbiased
estimator of the variance. Similar evaluation metric has been used
in \cite{11vldbreacheuncertaing}. Specifically, for a given seed
node $s$ in our experiments, we run all the estimators $\hat
{F}_s(\mathcal{G})$ $500$ times, thereby we can obtain $500$
estimating results: $\hat {F}^{(1)}_s(\mathcal{G}), \hat
{F}^{(2)}_s(\mathcal{G}), \cdots, \hat {F}^{(500)}_s(\mathcal{G})$.
An unbiased variance estimator of $\hat {F}_s(\mathcal{G})$ is given
by
\[\sum\nolimits_{i = 1}^{500} {(\hat {F}^{(i)}_s(\mathcal{G}) - \bar F_s(\mathcal{G}))^2 }
/499,
\]
where $\bar F_s(\mathcal{G}))$ denotes the mean of the $500$ various
estimating results.

\stitle{Parameter settings and the experimental environment}:
Without specifically stated, in all of our experiments, we set the
parameters as follows. For all estimators, we set the sample size
$N=1,000$. For the \bssi and \rssi estimators, we set $r=5$, and for
the \bssii and \rssii estimators, we set $r=50$. For the threshold
parameter $\tau$ in Algorithm~\ref{alg:rssi} and
Algorithm~\ref{alg:rssii}, we set $\tau=10$. All the experiments are
conducted on the Scientific Linux 6.0 workstation with 2xQuad-Core
Intel(R) 2.66 GHz CPU, and 4G memory. All algorithms are implemented
by GCC 4.4.4.

\subsection{Experimental Results}
\label{subsec:expresults}

For all the experiments, we randomly generate 1,000 seed nodes, and
the results are the average result over all the seeds. We report our
experimental results on random graph, FacebookLike, Condmat, and
DBLP dataset in Table~\ref{tbl:rvrg},
Table~\ref{tbl:rvfacebooklike}, Table~\ref{tbl:rvcondmat}, and
Table~\ref{tbl:rvdblp}, respectively.

From Table~\ref{tbl:rvrg}, among all the estimators, we can observe
that the \rssi-\bfs is the winner on the random graph dataset, the
\rssi-{\sl RM}, \rssii-{\sl RM}, and \rssii-\bfs estimators are
significantly better than the \rssi-{\sl RM} ($r=1$) estimator. The
specific \rssi-{\sl RM} ($r=1$) estimator outperforms the \bss
estimators, and all the \bss estimators are better than the \nmc
estimator. In particular, \rssi-\bfs reduces the relative variance
over the \nmc and \rssi-{\sl RM} ($r=1$) estimators by 386\% and
227\%, respectively.  \rssii-\bfs cuts the relative variance over
\nmc and \rssi-{\sl RM} ($r=1$) by 385\% and 226\%, respectively.
Both \rssi-{\sl RM} and \rssii-{\sl RM} estimators cut the relative
variance over the \nmc and the \rssi-{\sl RM} ($r=1$) estimators
more than 185\% and 91.4\%, respectively. For the \bss estimators,
their performance is worse than the \rssi-{\sl RM} ($r=1$)
estimator, but are significantly better than the \nmc estimator.
In addition, the running time of all the estimators are comparable.
These results consist with our analysis in
Section~\ref{sec:newestimatortypeI} and
Section~\ref{sec:newestimatortypeII}.

From Table~\ref{tbl:rvfacebooklike}, we can see that \rssii-\bfs
achieves the best relative variance on the FacebookLike dataset,
followed by \rssi-\bfs, \rssii-{\sl RM}, \rssi-{\sl RM}, \rssi-{\sl
RM} ($r=1$), the \bss estimators, and the \nmc estimator.  More
specifically, the \rssii-\bfs estimator reduces the relative
variance over the \nmc estimator and the \rssi-{\sl RM} ($r=1$)
estimators by 317\% and 133\%, respectively. The \rssi-\bfs
estimator reduces the relative variance over \nmc and \rssi-{\sl RM}
($r=1$) by 289\% and 117\%. Both \rssi-{\sl RM} and \rssii-{\sl RM}
estimators cut the relative variance over \nmc and \rssi-{\sl RM}
($r=1$) more than 231\% and 184\%, respectively. Similar to the
result on the random graph dataset, all the \bss estimators are
slightly worse than the \rssi-{\sl RM} ($r=1$) estimator but are
significantly better than the \nmc estimator.
Also, the running time of all the estimators are comparable because
the time complexities of all the estimators are $O(Nm)$. These
results confirm our analysis in the previous sections. Similar
results can be observed in the Condmat (Table~\ref{tbl:rvcondmat})
and DBLP datasets (Table~\ref{tbl:rvdblp}).

\begin{table}[t]\tbl{Results on random graph dataset
\label{tbl:rvrg}}{
\begin{tabular}{|l|c|c|}
\hline
{\bf Estimators} & {\bf Relative variance} & {\bf Running time (s)} \\
\hline
\nmc & 1.0000 &   0.3593 \\
\rssi-{\sl RM} ($r=1$) & 0.6723 & 0.3558 \\
\bssi-{\sl RM} & 0.9429 &  0.3497 \\
\bssi-\bfs & 0.8938 &  0.3748 \\
\rssi-{\sl RM} & 0.3397 &  0.3373 \\
\rssi-\bfs & \textbf{0.2056} &  0.3783 \\
\bssii-{\sl RM} & 0.9321 & 0.3633 \\
\bssii-\bfs & 0.9042&  0.3749 \\
\rssii-{\sl RM} & 0.3512 &  0.3716 \\
\rssii-\bfs & 0.2063 & 0.3847 \\
\hline
\end{tabular}}
\end{table}

\begin{table}[t]\tbl{Results on FacebookLike dataset
\label{tbl:rvfacebooklike}}{
\begin{tabular}{|l|c|c|c|}
\hline
{\bf Estimators} & {\bf Relative variance} & {\bf Running time (s)} \\
\hline
\nmc & 1.0000 & 0.2007  \\
\rssi-{\sl RM} ($r=1$) & 0.5585 &  0.2014 \\
\bssi-{\sl RM} & 0.8898 &  0.2331\\
\bssi-\bfs & 0.6819 &  0.2354\\
\rssi-{\sl RM} & 0.3023 &   0.2002\\
\rssi-\bfs & 0.2570 &  0.2010\\
\bssii-{\sl RM} & 0.6947 & 0.2250 \\
\bssii-\bfs & 0.6672 & 0.2284 \\
\rssii-{\sl RM} & 0.2786 &  0.2027\\
\rssii-\bfs & \textbf{0.2397} &  0.2037\\
\hline
\end{tabular}}
\end{table}

\begin{table}[t]\tbl{Results on Condmat dataset \label{tbl:rvcondmat}}{
\begin{tabular}{|l|c|c|}
\hline
{\bf Estimators} & {\bf Relative variance} & {\bf Running time (s)} \\
\hline
\nmc & 1.0000 &  1.2969  \\
\rssi-{\sl RM} ($r=1$) & 0.7950 & 1.2958 \\
\bssi-{\sl RM} & 0.9068 &  1.3043\\
\bssi-\bfs & 0.8531 &  1.3054\\
\rssi-{\sl RM} &  0.4883 &  1.2050 \\
\rssi-\bfs & \textbf{0.1971} & 1.2411 \\
\bssii-{\sl RM} &  0.8553 & 1.2513\\
\bssii-\bfs & 0.8421 & 1.3104 \\
\rssii-{\sl RM} & 0.4891 &  1.2256\\
\rssii-\bfs &  0.2120 &  1.2284\\
\hline
\end{tabular}}
\end{table}

\begin{table}[t] \tbl{Results on DBLP dataset \label{tbl:rvdblp}}{
\begin{tabular}{|l|c|c|}
\hline
{\bf Estimators} & {\bf Relative variance} & {\bf Running time (s)} \\
\hline
\nmc & 1.0000 &   8.5824 \\
\rssi-{\sl RM} ($r=1$) & 0.5375 & 8.6536 \\
\bssi-{\sl RM} & 0.9170 &  8.6292 \\
\bssi-\bfs & 0.8373 &  8.8173 \\
\rssi-{\sl RM} & 0.2100 &  8.3835 \\
\rssi-\bfs & 0.1918 &  8.5933 \\
\bssii-{\sl RM} & 0.9449 & 8.8825 \\
\bssii-\bfs & 0.7997&  9.1305 \\
\rssii-{\sl RM} &  0.2003 &  8.6840 \\
\rssii-\bfs & \textbf{0.1821} & 8.7052 \\
\hline
\end{tabular}}
\end{table}

To summarize, \rssi-\bfs and \rssii-\bfs achieve the best relative
variance, and they reduce the relative variance over the existing
estimators several times. The \rss estimators are better than the
\bss estimators. The \bss/\rss estimators with the \bfs
edge-selection strategy are better than the \bss/\rss estimators
with the random edge-selection strategy. All of our \rss estimators
outperform the \rssi-{\sl RM} ($r=1$) estimator. The proposed \bss
estimators are slightly worse than the \rssi-{\sl RM} ($r=1$)
estimator, but still significantly outperform the \nmc estimator.
The running time of all the estimators are comparable.

\stitle{Scalability}: In order to study the scalability of various
estimators, we generate synthetic probabilistic graphs $\mathcal{G}$
with nodes ranging from 200,000 (200k) to 800,000 and the edges
ranging from 800,000 to 3,200,000 (3.2m) according to the ER random
graph model. And the probability of each edge is randomly generated
according to a [0, 1] uniform distribution. Also, for each
estimator, we set the sample size $N$ to 1,000.
Table~\ref{tbl:scalability} shows the running time of different
estimators on four large synthetic probabilistic graphs. As can be
seen in Table~\ref{tbl:scalability}, the running time increases as
the size of the graph increases. In general, all the estimators
achieve comparable running time, and they have linear growth w.r.t.
the graph size. These results consist with the complexities of our
estimators, i.e.\ $O(Nm)$.

\begin{table}[t] \tbl{Scalability: Running time on synthetic graphs. Here the two numbers
in the 2nd-5th columns (eg.\ 200k/800k) indicate the numbers of
nodes and edges respectively \label{tbl:scalability}}{
\begin{tabular}{|l|c|c|c|c|}
\hline
{\bf Time (s)} & 200k/800k & 400k/1.3m & 600k/1.6m & 800k/3.2m \\
\hline
\nmc & 26.0820 & 156.9600 & 289.7720 & 365.0280 \\
\bssi-{\sl RM} & 25.2090 & 159.1990 & 281.6810 & 343.0350 \\
\bssi-\bfs & 27.2120 & 169.6120 & 286.2180 & 368.0910 \\
\rssi-{\sl RM} & 23.3430 & 143.6700 & 264.9790 & 342.3920 \\
\rssi-\bfs & 25.2090 & 169.6120 & 286.2180 & 344.0180 \\
\bssii-{\sl RM} & 26.1450 & 161.4100 & 287.1500 &  371.4770\\
\bssii-\bfs & 29.5760 & 162.3930 & 290.9340 & 374.6830\\
\rssii-{\sl RM} &  26.4440 & 156.8120 & 270.6670 & 363.1590 \\
\rssii-\bfs & 26.4990 & 162.7940 & 271.1370 & 365.9630 \\
\hline
\end{tabular}}
\end{table}

\stitle{Effect of parameter $r$}: We study the effectiveness of the
parameter $r$ in our proposed estimators on Condmat dataset. Similar
results can be observed from other datasets. Fig.~\ref{fig:reffectI}
and Fig.~\ref{fig:reffectII} show the relative variance of our
type-I and type-II estimators w.r.t.\ various $r$. As can be seen in
Fig.~\ref{fig:reffectI}, the \bssi estimators exhibit similar
relative variance over different $r$ values. However, the relative
variance of the \rssi-{\sl RM} estimator decreases as the $r$
increases when $r \le 5$, and otherwise it increases as the $r$
increases. For the \rssi-\bfs estimator, the relative variance
decreases as $r$ increases, and when $r \ge 5$ the descent rate is
very small, and the curve tends to be smooth. Based on this
observation, $r=5$ is the best choice, which is used in the previous
experiments.
For the type-II estimators, we test the parameter $r$ from 10 to 70,
and the results (Fig.~\ref{fig:reffectII}) show that all of our
type-II estimators except \rssi-\bfs are not very sensitive w.r.t.\
the parameter $r$. As an exception, the relative variance of the
\rssi-\bfs estimator decreases as the $r$ increases when $r \le 50$,
and when $r \ge 50$ the the curve tends to be smooth. Therefore,
$r=50$ is a good choice. In our previous experiments, we set $r$ to
50. Table~\ref{tbl:rtimerssI} and Table~\ref{tbl:rtimerssII} report
the running time of type-I estimators and type-II estimators under
different $r$ values.  We can see that the running time of both
type-I estimators and type-II estimators are comparable.

\begin{figure}[t] 
\begin{center}
\includegraphics[width=0.8\hsize]{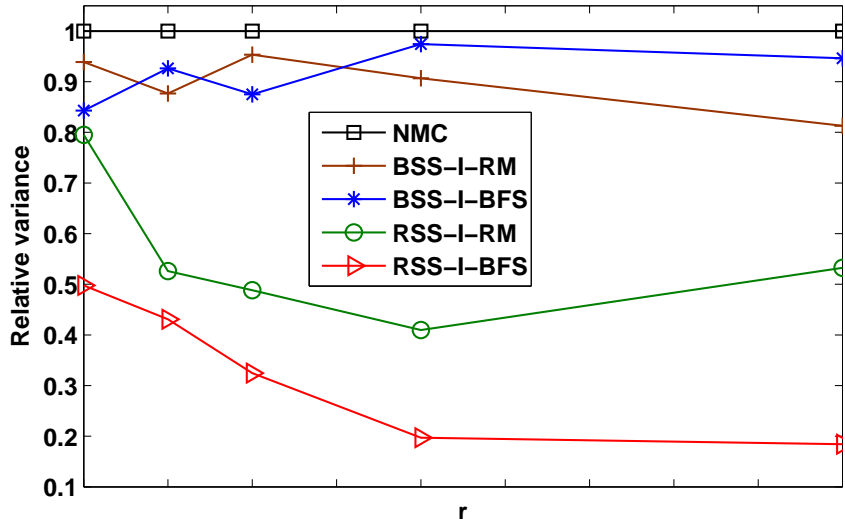}
\vspace*{-1em} \caption{Effect of $r$ of \bssi/\rssi estimators.}
\label{fig:reffectI}
\end{center}
\end{figure}

\begin{figure}[t] 
\begin{center}
\includegraphics[width=0.8\hsize]{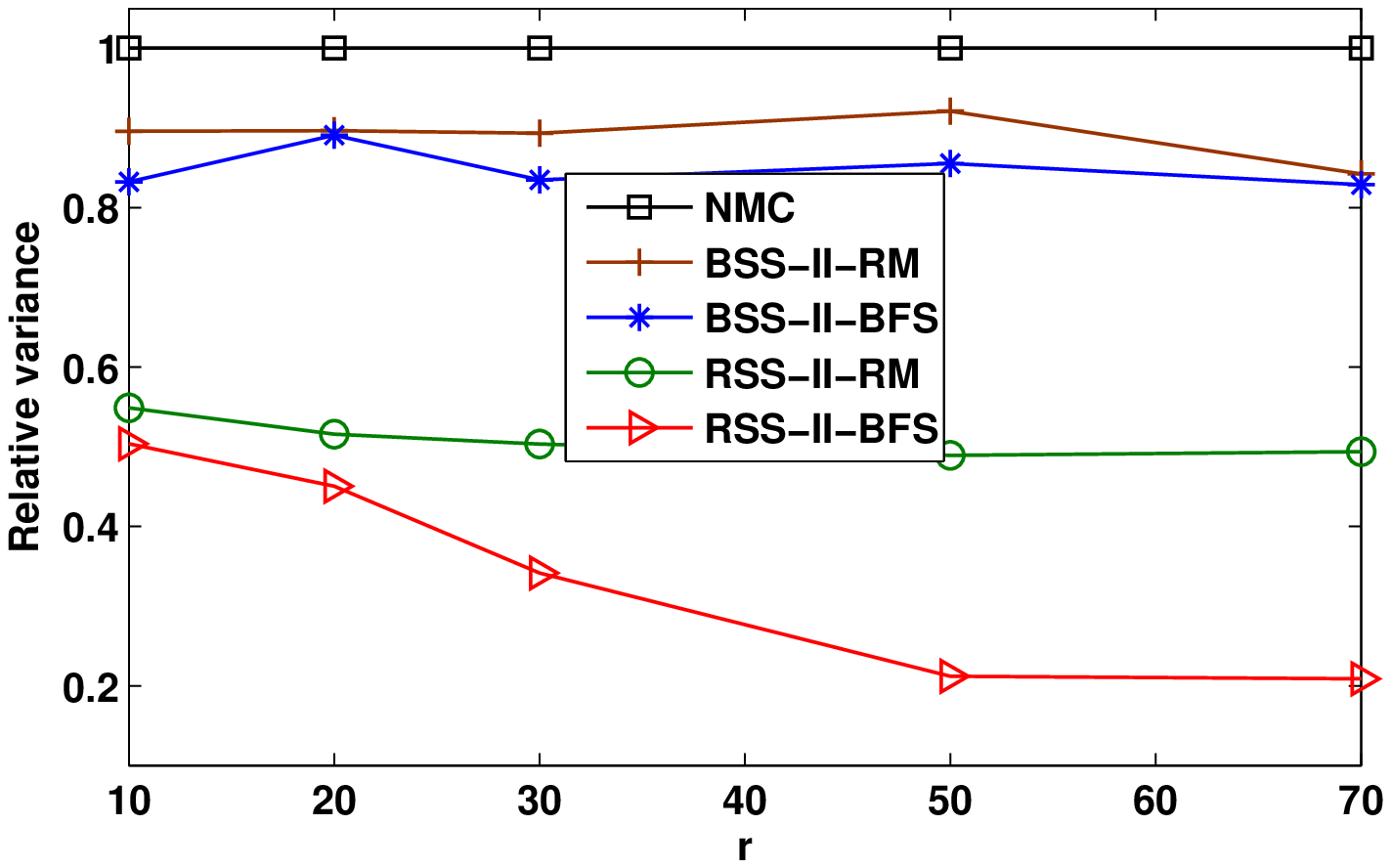}
\vspace*{-0.8em} \caption{Effect of $r$ of \bssii/\rssii
estimators.} \label{fig:reffectII}
\end{center}
\end{figure}

\begin{table}[t]\tbl{\bssi/\rssi estimators: Running time vs $r$
\label{tbl:rtimerssI}}{
\begin{tabular}{|l|c|c|c|c|c|}
\hline
{\bf Time (s)} & $r=1$ & $r=2$ & $r=3$ & $r=5$ & $r=10$ \\
\hline
\bssi-{\sl RM} & 1.2642 & 1.2705 & 1.2770 & 1.3043 & 1.2986 \\
\bssi-\bfs & 1.2731 & 1.2791 & 1.2798 & 1.3054 & 1.3686 \\
\rssi-{\sl RM} & 1.2082 & 1.1993 & 1.1810 & 1.2050 & 1.1172 \\
\rssi-\bfs & 1.2158 & 1.2140 & 1.2189 & 1.2411 & 1.1833\\
\hline
\end{tabular}}
\end{table}

\begin{table}[t]\tbl{\bssii/\rssii estimators: Running time vs $r$
\label{tbl:rtimerssII}}{
\begin{tabular}{|l|c|c|c|c|c|}
\hline
{\bf Time (s)} & $r=10$ & $r=20$ & $r=30$ & $r=50$ & $r=70$ \\
\hline
\bssii-{\sl RM} & 1.2515 & 1.2502 & 1.2511 & 1.2513 & 1.2524 \\
\bssii-\bfs & 1.2579 & 1.2719 & 1.2836 & 1.3104 & 1.3447 \\
\rssii-{\sl RM} & 1.2279 & 1.2246 & 1.2162 & 1.2256 & 1.2092 \\
\rssii-\bfs & 1.2358 & 1.2278 & 1.2258 & 1.2284 & 1.2295 \\
\hline
\end{tabular}}
\end{table}

\stitle{Effect of sample size}: As shown in the previous
experiments, the \rssi-\bfs and the \rssii-\bfs estimators are the
best two estimators. Here we study how sample size affects the
estimating accuracy of these two estimators on the Condmat dataset.
Similar results can be observed on the other dataset.
Fig.~\ref{fig:samplesize} shows the relative variance of the
estimators under various sample size. As can be observed in
Fig.~\ref{fig:samplesize}, the curves of \rssi-\bfs and \rssii-\bfs
estimators are very smooth, which indicate that the relative
variance of both \rssi-\bfs and \rssii-\bfs estimators are robust
w.r.t.\ the sample size.

\begin{figure}[t] 
\begin{center}
\includegraphics[width=0.8\hsize]{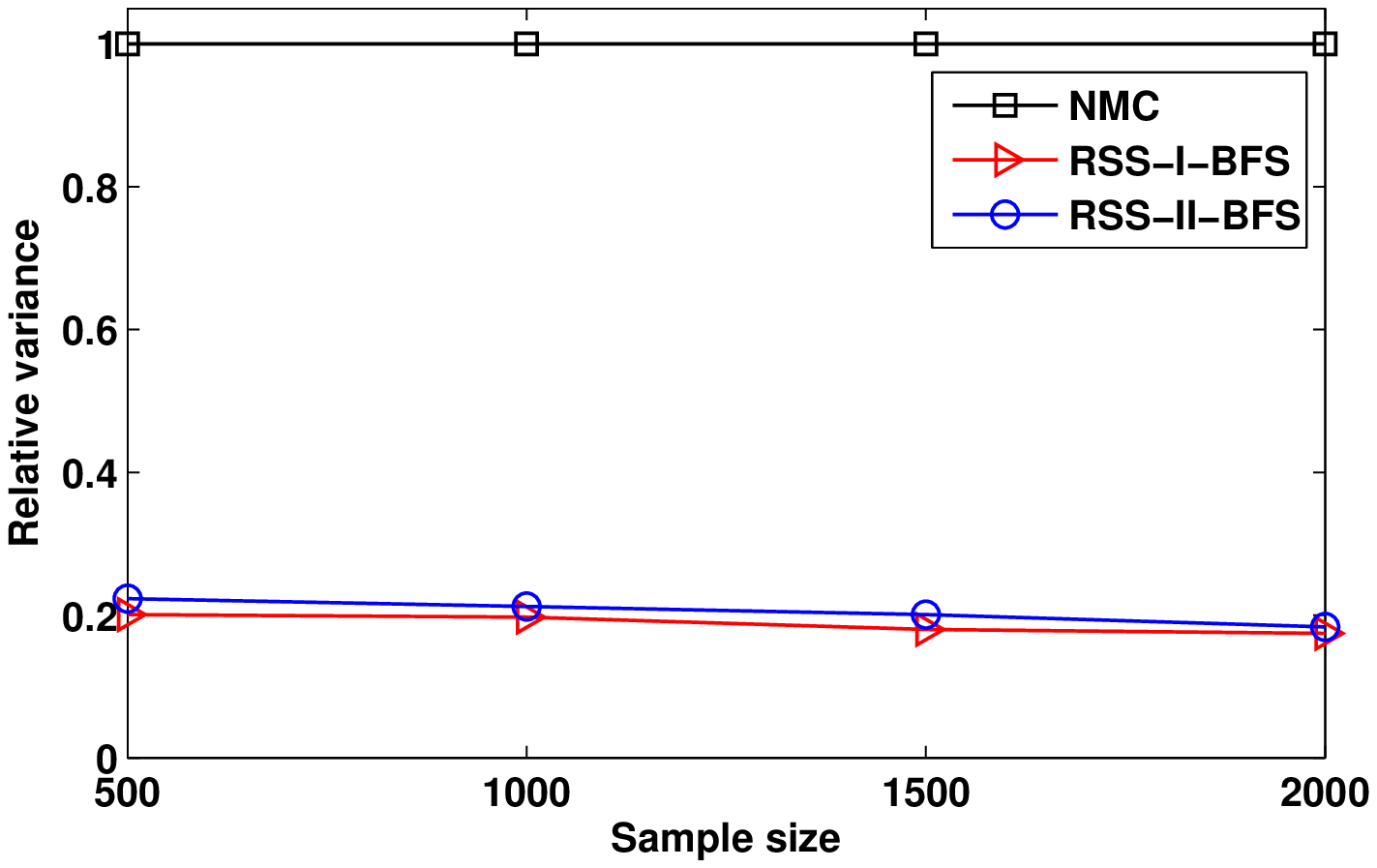}
\vspace*{-0.8em} \caption{Relative variance vs sample size.}
\label{fig:samplesize}
\end{center}
\end{figure}

%

%


\section{Related work}
\label{sec:rlwork}

After the seminal work by Kempe, et al. \cite{03kddinfluence},
influence maximization in social networks has recently attracted
much attention in data mining and social network analysis research
communities
\cite{05icalpinfluence,07kddoutbreak,09kddinfluence,10kddinfluence,10wsdmlearninginfluence,11sdminfluence,12vldbdatadriveninfluence}.
A crucial subroutine in influence maximization is the influence
function evaluation to which the influenceability estimation problem
presented in this paper is closely related.
In the following, we first review some notable work on influence
maximization problem, and then discuss the existing work on
influence function evaluation. In \cite{07kddoutbreak}, the authors
study the influence maximization problem under the context of water
distribution and blogosphere monitoring. They propose a so-called
CELF framework for optimizing the influence maximization algorithms.
To further accelerate the influence maximization algorithms, Chen,
et al. in \cite{09kddinfluence} propose a scalable algorithm by
sampling $N$ possible graphs and estimating the influence spread of
all vertices on each possible graph at one time. Subsequently, the
same authors propose a series of scalable algorithms in
\cite{10kddinfluence} and \cite{11sdminfluence} for influence
maximization by developing the heuristic vertices-selection
strategies on unsigned and signed networks, respectively. Recently,
Goyal, et al. in
\cite{10wsdmlearninginfluence,12vldbdatadriveninfluence} consider
the problem of learning the influence probabilities, and study the
influence maximization from a data-driven perspective. Note that all
the mentioned methods focus on the influence maximization problem.
For the influence function evaluation problem, Kempe, et al. firstly
pose it as an open problem in \cite{05icalpinfluence}. Then, Chen,
et al. in \cite{10kddinfluence} show that the influence function
evaluation problem is \#P complete. Given the hardness of the
problem, most of the existing work for this problem, such as
\cite{03kddinfluence,07kddoutbreak,09kddinfluence,10kddinfluence},
are based on the Naive Monte-Carlo (\nmc) sampling. In this paper,
we study the influenceability evaluation problem and develop more
accurate \rss estimators for estimating the influenceability, and
our algorithms can also be used for influence function evaluation.

Our work is also related to the uncertain graph mining. Recently,
uncertain graphs mining have been attracted increased interest
because of the increasing applications in biological database
\cite{06dilsuncertaingraph}, network routing
\cite{07infocommuncertainnet}, and influence networks
\cite{12vldbdatadriveninfluence}. There are a large body of works
have been proposed in the literature. Notable work includes finding
the reliable subgraph in a large uncertain graph
\cite{08dmkdprobgraph,11kddreliablesugg}, frequent subgraph mining
in uncertain graph database
\cite{10kdduncertaingraph,10tkdeuncertaingraph}, subgraph search in
large uncertain graph \cite{11vldbuncertaingraph}, K-nearest
neighbor search in uncertain graph \cite{10vldbuncertaingraph}, and
distance constraint reachability computation in uncertain graph
\cite{11vldbreacheuncertaing}. In general, all the mentioned
uncertain graph mining problems are shown to be \#P-complete, and
thereby finding the exact solution is intractable in large uncertain
graphs. Consequently, most existing work, such as
\cite{10vldbuncertaingraph} and \cite{11kddreliablesugg}, are based
on \nmc sampling. Basically, the \nmc sampling based methods lead to
a large variance, thus reduce the performance of the algorithms.
Recently, Jin, et al. in \cite{11vldbreacheuncertaing} propose a
recursive stratified sampling method for distance-constraint
reachability computation on uncertain graph, although they do not
claim their method is a stratified sampling. It is important to note
that their method is a very special case of our \rssi algorithm. In
their method, they select only one edge for stratification at a
time, and then recursively perform this procedure. Unlike their
algorithm, first, we develop a generalized algorithm (\rssi) that
selects $r$ edges for stratification. Second, unlike their
reachability problem, here we study the influenceability evaluation
problem using the \rssi sampling. Moreover, in our work, we also
develop another \rss estimator, i.e. \rssii estimator. Note that all
of our \rss estimators can also be applied into the
distance-constraint reachability computation problem.

In addition, our work is related to the network reliability
estimation problem, where a network is modeled as an uncertain graph
and the goal is to estimate some reliability metrics of the network
\cite{86trmccompare,99bookchnetreliability}. There are many work on
this topic in the last five decades. Surveys can be found in
\cite{87booknetreliable,99bookchnetreliability}.

Below, we review the Monte-Carlo algorithms for network reliability
estimation. Kumamoto, et al. \cite{77trboundreliability} propose an
efficient Monte-Carlo algorithm by exploiting the bound of the
reliability metric. Fishman \cite{86ormcbounds} proposes a more
generalized Monte-Carlo algorithm based on such bound techniques for
reliability estimation. Subsequently, Fishman \cite{86trmccompare}
compares four Monte-Carlo algorithms for network reliability
estimation problem. Cancela, et al. in \cite{03trrvrnetreliable}
propose a recursive variance-reduction algorithm for network
reliability estimation. Note that all the mentioned Monte-Carlo
algorithms are tailored for the network reliability estimation
problem, and the reliability measure is typically a Boolean metric
thus they cannot be used in our problem.

\section{Conclusions}
\label{sec:concl}

In this paper, we focus on the influenceability evaluation problem,
which is a fundamental issue for influence analysis in social
network. This problem is known to be \#P-complete, and the only
existing algorithm is based on the Naive Monte-Carlo (\nmc)
sampling. To reduce the variance of the \nmc estimator, we propose
two basic stratified sampling (\bss) estimators. Furthermore, based
on our \bss estimators, we present two recursive stratified sampling
(\rss) estimators. We conduct comprehensive experiments on one
synthetic and three real datasets, and the results confirm that our
\rss estimators reduce the variance of the \nmc estimator by several
times. There are several future directions that deserve further
investigation. First, most of our estimators except the \rss
estimators with BFS edge selection do not take the graph structural
information into account. In our experiments, the \rss estimators
with BFS edge selection are shown much better performance than the
\rss estimators with random edge selection. A promising direction is
to exploit the graph structural information to develop more
efficient and more accurate estimators for influenceability
evaluation. Second, our estimation techniques are quite general. For
many uncertain graph mining problems, such as shortest path
\cite{10vldbuncertaingraph}, reachability
\cite{11vldbreacheuncertaing}, and reliable subgraph discovery
\cite{11kddreliablesugg}, our estimators can be directly used. For
these problems, we only need to replace the $\phi_s(G_P)$ to other
quantities, such as the length of the shortest path, the
reachability function between two nodes, and the reliable subgraph
metric. Most of these uncertain graph mining problems are based on
\nmc. Another promising future direction is to apply our estimation
techniques to these problems.



%

\bibliographystyle{acmsmall}
\bibliography{tkddinfluence}

%
%
%
%
%
%
%

\end{document}